\documentclass[a4paper, 11pt]{article}
\usepackage[utf8]{inputenc}
\usepackage{jheparxiv}
\usepackage{amsmath,mathrsfs,bm,amsthm}
\DeclareUnicodeCharacter{0360}{\eth}
\usepackage[final]{showkeys}
\usepackage{comment}
\usepackage{tikz-cd}
\usepackage{bm}
\usepackage{svg}
\usepackage[normalem]{ulem}
\usepackage[shortlabels]{enumitem}
\usepackage[normalem]{ulem}

\theoremstyle{definition}
\newtheorem{Theorem}{Theorem}
\theoremstyle{definition}
\newtheorem{Proposition}{Proposition}
\newtheorem{Lemma}{Lemma}
\theoremstyle{definition}

\theoremstyle{definition}

\definecolor{arXiv}{named}{blue}

\newcommand{\IC}[0]{\mathbb{C}}
\newcommand{\IR}[0]{\mathbb{R}}

\newcommand{\F}{\mathscr{F}}
\newcommand{\T}{\mathscr{T}} 
\newcommand{\PT}{\textrm{P}\mathcal{T}} 
\newcommand{\HST}{M_{HS}}
\newcommand{\M}{\mathcal{M}} 
\newcommand{\h}{\bm{h}}

\title{Higher-spin self-dual gravity from holomorphic planes in twistor space}

\author{Nicolas Boulanger$^{\dagger}$, Yannick Herfray$^{\ddagger \star}$, Lionel Mason$^*$,
 Noémie Parrini$^{\dagger}$\footnote{Research Fellow of the F.R.S.-FNRS (Belgium).}}

\affiliation{$^{\dagger}$Physique de l'Univers, Champs et Gravitation, 
 Universit\'e de Mons -- UMONS,\\Place du Parc 20, 7000 Mons, Belgium}

\affiliation{$^{\ddagger}$Institut Denis Poisson UMR 7013, Université de Tours, 
Parc de Grandmont, 37200 Tours, France}
\affiliation{$^\star$ Institut universitaire de France (IUF)}

\affiliation{$^{*}$ The Mathematical Institute, University of Oxford, United Kingdom}

\emailAdd{nicolas.boulanger@umons.ac.be, yannick.herfray@univ-tours.fr, lmason@maths.ox.ac.uk, noemie.parrini@umons.ac.be}

\abstract{We prove a `nonlinear graviton theorem' for higher-spin self-dual 
gravity. We consider small deformations of the complex structure of the 
non-projective twistor space that are bounded in a specified region near the origin and investigate the 
space $\HST$ of holomorphically embedded complex planes $\mathbb{C}^2$ that intersect the origin.
We show  that this space is an infinite dimensional complex 
manifold with a canonical projection onto a four-dimensional holomorphic self-dual 
spacetime $\M$, and discuss the geometry induced on this new higher-spin space. 
Solutions of higher-spin self-dual gravity are then obtained by choosing an embedding of spacetime 
$\M$ into higher-spin space $\HST$, with higher-spin symmetries arising from the 
different choices of embedding.  
Integrability of the theory is manifested in the form of a Lax pair for the system that we present.   We 
conjecture that chiral higher-spin gravity can similarly be realized by considering 
deformations that are unconstrained at the origin.}

\begin{document}
\maketitle

\newpage
\section{Introduction}

Interacting theories of massless higher-spin fields are notoriously difficult 
to construct. Among the few known examples, chiral higher-spin gravity 
(chiral HiSGra) \cite{Metsaev:1991nb,Metsaev:1991mt,Ponomarev:2016lrm,Skvortsov:2018jea,Skvortsov:2020wtf}   
(see also \cite{Skvortsov:2020gpn,Skvortsov:2022syz,Sharapov:2022faa,Sharapov:2022awp,Sharapov:2022wpz}) 
is expected to be integrable and is thus especially worth of interest. 
Chiral HiSGra admits several interesting truncations \cite{Ponomarev:2017nrr,Krasnov:2021nsq,Monteiro:2022xwq}. 
In particular, higher-spin self-dual gravity and higher-spin self-dual 
Yang-Mills admit covariant spacetime actions \cite{Krasnov:2021nsq} and can be 
thought of as higher-spin extensions of self-dual gravity and 
self-dual Yang-Mills theory, respectively.

Bearing in mind the analogy with self-dual Yang-Mills and self-dual gravity, it is natural to expect 
that the integrability properties of these theories should derive from a 
twistor realization. For higher-spin-self-dual Yang-Mills such realizations were discussed in 
\cite{Adamo:2022lah,Herfray:2022prf}. The Ward correspondence of 
\cite{Herfray:2022prf} identifies solutions to higher-spin-self-dual Yang-Mills with holomorphic 
bundles on the \emph{non projective} twistor space\footnote{Here and 
everywhere in this article $\mathbb{C}^ n_{0}$ stands for $\mathbb{C}^n$ 
minus the origin.} $\mathbb{T}$, identified with a region of $ \mathbb{C}^4_0\setminus \mathbb{C}^2_0$. 
This should be be contrasted with the usual Ward correspondence \cite{Ward:1977ta} 
that, along with the usual results of twistor theory, make use of a projective twistor space 
$\mathbb{PT}$ inside $  \mathbb{C}\textrm{P}^3 \setminus \mathbb{C}\textrm{P}^1$. 
A similar strategy for a twistor action for chiral HiSGra was 
constructed in \cite{Mason:2025pbz} whose  field equations amount to the 
integrability of a deformed \emph{Cauchy-Riemann}-structure on a region 
in the unit sphere inside $\mathbb{T}$, rather than on 
$\mathbb{PT}$ as in the usual twistor action for self-dual gravity \cite{Mason:2007ct}.

Several open questions remain. 
The complete understanding of the higher-spin generalization of the nonlinear 
graviton theorem \cite{Penrose:1976js} is still missing, 
and the work \cite{Mason:2025pbz} strongly relies on a perturbative approach 
and on the use of Euclidean signature structure, 
whereas these are likely to be  unnecessary restrictions. 
In the present work, as a first step towards resolving theses questions, 
we present a nonlinear graviton theorem for higher-spin self-dual gravity with zero cosmological 
constant, without the need to assume any reality 
structure. We will prove the following.

\begin{Theorem}\label{theorem intro}
There is a one-to-one correspondence between, small but finite, solutions of 
higher spin self-dual gravity with zero cosmological constant up 
to gauge and, small but finite, deformations of the complex structure of 
$\mathbb{T}$ in a region of $\mathbb{C}^4_0\backslash \mathbb{C}^2_0$ that
(1) preserve the holomorphic fibration $\mathbb{T} \to \mathbb{C}^2_0\,$, 
(2) preserve the holomorphic Poisson structure along the fibre, 
(3) admits an  Euler vector field $E$ that is $(1,0)$ but \emph{not} required to be holomorphic, and 
(4) have  Dolbeault $\bar \partial$-operator that is bounded near the origin.
\end{Theorem}

The main weaking of Penrose's classical theorem \cite{Penrose:1976js}, 
is in the third assumption: the  holomorphicity of the Euler vector field of Penrose's construction  
is here weakened to the assumption of simply being $(1,0)$ but no longer holomorphic. 
As a result, the generic complex deformations that we will consider 
cannot be interpreted as complex deformations of $\mathbb{PT}$, 
and higher spin fields will arise from this extra freedom.  
The fourth assumption would automatically follow from the classical ones; 
as we shall see, it effectively prevents the appearance of negative 
helicity higher-spin fields. We conjecture that lifting this fourth 
requirement would allow, together with the use of the Moyal product as in 
\cite{Mason:2025pbz}, to describe solutions of chiral HiSGra. 
See also the further discussion at the end of this article. Finally, integrability is manifested in our construction first by the existence of a Lax pair, which follows from the generalisation of a result from \cite{Herfray:2022prf}, and second, from the nonlinear graviton construction of Theorem \ref{theorem intro}. \\

The article is organised as follows. 
In section \ref{section: Higher-spin self-dual gravity} we review, from 
\cite{Krasnov:2021nsq}, the field equations and action for higher-spin self-dual gravity. 
In the following section \ref{section: Twistor space of HS-SDGR}, 
we explain how the twistor space of higher-spin self-dual gravity with zero cosmological constant 
is obtained, which is essentially a generalization of the results of 
\cite{Herfray:2022prf} in the complexified rather than Euclidean context. 
However, due to the fact that we are here working with zero cosmological 
constant, higher spin gauge symmetry can be given a complete treatment. 
Section \ref{section: Nonlinear HS graviton theorem} contains our main 
results: we consider a deformation of the complex structure on $\mathbb{T}$, 
along the line of \cite{Mason:2025pbz} but without the need to assume reality 
conditions, and prove that in the linearized theory it generically describes 
a full tower of higher-spin fields. We then study the space $\HST$ of 
holomorphic planes in the curved twistor space and show that it is an infinite 
dimensional complex manifolds $\HST \to \M$ which fibers over a 4-dimensional 
self-dual spacetime $\M$. We then prove that a choice of embedding $\M\hookrightarrow \HST$ yields a solution to higher-spin self-dual gravity. In section \ref{section: Higher-spin space geometry and higher-spin symmetries}, we close the loop by proving that two solutions obtained via two different embeddings are related by a higher-spin symmetry and that all higher-spin symmetries can be realized in that way. In the last section, we derive the Plebanski potentials of higher-spin self-dual gravity from our twistor space description and show that there is precise room for obtaining interactions with negative helicity fields.

\section{Higher-spin self-dual gravity}\label{section: Higher-spin self-dual gravity}

We here review from \cite{Krasnov:2021nsq}, the field equations and action for higher-spin  self-dual  gravity with zero cosmological constant.\\

Let $\M$ be a $4d$ manifold with coordinates $x^{\mu}$ and let us note 
$\alpha\in 0,1$ some global spinor indices equipped with a canonical volume form 
$\epsilon_{\alpha\beta}\,$. 
Everywhere $\alpha(s)= (\alpha_1 \cdots \alpha_s)$ is a multi-index notation 
standing for the symmetrised products of such indices. 
Our convention is that if the same letter is used, symmetrisation 
is implicit e.g. $\alpha(s)\alpha(s')=\alpha(s+s')\,$. 
The fields for higher-spin self-dual gravity are
\begin{align}
    \Psi_{\alpha(l+4)}&, & A^{\alpha(l+2)} = 
    {\rm d}x^{\mu}A_{\mu}^{\alpha(l+2)}, \hspace{3cm}\forall\; l\in 
    \mathbb{N}^*\;,
\end{align}
here $l =2s-4$, $\Psi_{\alpha(l+4)}$ is a field of helicity $ - s$ 
and $A_{\mu}^{\alpha(l+2)}$ is a field of helicity $+s\,$.

The spacetime action for higher-spin self-dual gravity is \cite{Krasnov:2021nsq}
\begin{align}\label{Lagrangian of SD HS gravity}
 	S[\Psi, A] &= \frac{1}{2} \int \sum_{l\geq 0} \frac{1}{(l+4)!}
    \Psi_{\alpha(l+4)} \,(F\wedge F)^{\alpha(l+4)}\;,
 \end{align}
where 
\begin{align}\label{SD HS gravity: def of FwF}
    F^{\alpha(k)} &:= {\rm d} A^{\alpha(k)}, & (F\wedge F)^{\alpha(k)} := 
    \sum_{l+m=k} F^{\alpha(l)} \wedge F^{\alpha(m)}\;.
\end{align}

Varying with respect to $\Psi^{\alpha(s)}$ this action one obtains the 
field equations
\begin{align}\label{Field equations of SD HS gravity}
	\sum_{k+l=s} F^{\alpha(k)} \wedge F^{\alpha(l)} &= 0\;.
\end{align}
where we recall that it is our conventions that indices with the same 
name are implicitly symmetrized.

The field equations have the following (infinitesimal) gauge freedom:
\begin{align}\label{HS SD gravity: gauge}
\begin{alignedat}{2}
      \delta A^{\alpha(s)} &= {\rm d} \chi^{\alpha(s)},\\
    \delta A^{\alpha(s)} & = {\rm d}x^{\mu} \left( \sum_{k+l=s} 
    \eta^{\alpha(k) \nu}\, F^{\alpha(l)}_{\mu\nu} \right)\;.
\end{alignedat}
\end{align}

\section{Twistor space of higher-spin self-dual gravity}\label{section: Twistor space of HS-SDGR}

We will first show how to obtain the twistor space associated to a solution 
of the higher-spin self-dual gravity field equations. The present construction is an adaptation, 
to complexified spacetime solutions of higher-spin self-dual gravity with $\Lambda=0$, of the 
result of \cite{Herfray:2022prf} (which considered solutions with 
$\Lambda\neq0$ and in Euclidean signature).

\subsection{Correspondence space}

The correspondence space $\F$ is defined as the product 
$\F = \M \times \IC^2_0 $ of the spacetime and the space $\IC^2_0$ 
of non zero spinors $\lambda^{\alpha}$, accordingly one will take 
\begin{equation}
   (x^{\mu}, \lambda_\alpha) \in  \;\M \times \IC^2_0 \,=\, \F
\end{equation}
 as our coordinates on the correspondence space. 
 This is a product manifold with two projections 
 $\pi_1 : \F \to \M$ and $\pi_2 : \F \to \IC^2_0$ --  which is in 
 line with the fact that our indices $\alpha$ are global. 
 This also means that the exterior derivative naturally is the sum 
 \begin{equation}
     {\rm d} = {\rm d}_x +  {\rm d}_{\lambda}
 \end{equation}
 of the spacetime derivative (``horizontal'') 
 ${\rm d}_x = {\rm d}x^{\mu} \frac{\partial}{\partial x^{\mu}}$ 
 and the vertical derivative 
 ${\rm d}_{\lambda} = {\rm d}\lambda^{\alpha} \frac{\partial}{\partial 
 \lambda^{\alpha}}\,$.

 The fields of higher-spin self-dual gravity naturally define a horizontal 1-form on 
 $\M \times \IC^2_0 \,=\, \F$
 \begin{equation}\label{A on correspondence space}
     A := {\rm d}x^{\mu} \left( \sum_{l\geq2} A^{\alpha(l)}_{\mu} 
     \lambda_{\alpha(l)} \right)
 \end{equation}
which is holomorphic along $\IC^2_0\,$. 
The other way round, any horizontal 1-form on $\F$ which is holomorphic 
along $\IC^2_0$ and vanishes at the origin to first order can be written 
as \eqref{A on correspondence space} (at least in a neighborhood of $\lambda=0$) 
since, by Hartogs's theorem, it must in fact be holomorphic on $\mathbb{C}^2$.

Introducing the horizontal 2-form
\begin{equation}\label{F on correspondence space}
    \begin{alignedat}{2}
    F &:=  {\rm d}_x A \\
    &=   \frac{1}{2}{\rm d}x^{\mu}\wedge {\rm d}x^{\nu}\; \left( \sum_{l\geq2} 
    F^{\alpha(l)}_{\mu\nu} \lambda_{\alpha(l)} \right)
    \end{alignedat}
\end{equation}
where $F^{\alpha(l+2)} =  \frac{1}{2}\, {\rm d}x^{\mu}\wedge {\rm d}x^{\nu}\; F^{\alpha(l+2)}_{\mu\nu}$ 
is the spacetime 2-form \eqref{SD HS gravity: def of FwF} appearing 
in the Lagrangian \eqref{Lagrangian of SD HS gravity}. Now, since
\begin{align}
\begin{alignedat}{2}
    F \wedge F  &= \left( \sum_{k\geq2} F^{\alpha(k)} \lambda_{\alpha(k)} \right) 
    \wedge \left( \sum_{l\geq2} F^{\alpha(l)} \lambda_{\alpha(l)}\right) \\
    &= \sum_{s\geq4}  \left( \sum_{k+l=s} F^{\alpha(k)}\wedge F^{\alpha(l)}  \right) 
    \lambda_{\alpha(s)}\;,
    \end{alignedat}
\end{align} 
the field equations \eqref{Field equations of SD HS gravity} 
of higher-spin self-dual gravity are in fact equivalent to
\begin{equation}\label{Field equation for F in correspondence space}
    F \wedge F = 0\;.
\end{equation}

\paragraph{Gauge symmetry in correspondence space}

The infinitesimal gauge symmetry \eqref{HS SD gravity: gauge} has a natural nonlinear realization in the correspondence space in terms of horizontal, holomorphic, diffeomorphisms :
\begin{align}\label{Diffeo in correspondence space}
	\Phi \left|\begin{array}{ccc}
	     \M \times \mathbb{C}_0^2 & \longrightarrow &
	\M \times \mathbb{C}_0^2 \\
	    (x^\mu, \lambda_\alpha) &\longmapsto& (\varphi^{\mu}(x, \lambda_\alpha), \lambda_\alpha)
	\end{array}
	\right.
\end{align}
with $\frac{\partial}{\partial \bar\lambda}\varphi^{\mu} = 0$. 
The action, via pull-back, of these diffeomorphisms
\begin{equation}
 \varphi^{\mu}(x, \lambda_\alpha) = \sum_{k} \varphi^{\mu \alpha(k)}(x) \lambda_{\alpha(k)}   
\end{equation}
gives 
$\Phi^*A= {\rm d} \varphi^{\mu} A_{\mu}(\varphi(x,\lambda),\lambda)
+{\rm d}\lambda^\alpha A_\alpha\,$. 
Projecting the resulting form on the horizontal directions\footnote{One 
can consider a one-form $A$ with both vertical and horizontal components, 
and from it, its two form $F={\rm d}A\,$.
The equations $F\wedge F=0$ can be solved explicitly, 
and again imply that that the horizontal part $F_{\mu\nu}$ is simple. 
The purely vertical 
components $F_{\alpha\beta}=\xi_{\alpha\dot\alpha}\xi_{\beta\dot\beta}
\epsilon^{\dot{\alpha}\dot{\beta}}$ of the total two-form involve a 
vertical frame $\xi_{\alpha\dot\alpha}$ that has received no 
interpretation so far.} 
yields an action
\begin{equation}\label{eq action of gauge symmetry on A}
    A= {\rm d}x^{\mu} A_{\mu}(x,\lambda) \qquad \mapsto \qquad 
    \Phi\cdot A := {\rm d}_x \varphi^{\mu} A_{\mu}(\varphi,\lambda)\;,
\end{equation}
which preserves the form \eqref{A on correspondence space} as well as the field equations \eqref{Field equation for F in correspondence space}.

The corresponding infinitesimal transformation of the potential one form \eqref{A on correspondence space} under this diffeomorphism is
\begin{align}
	\delta_{\eta} A = {\rm d}_x(\iota_\eta A)+ \iota_\eta({\rm d}_x A) 
    = {\rm d}_x (\iota_\eta A) + \iota_\eta {F}\;,
\end{align}
with $\eta= \sum_{k} \eta^{\mu \alpha(k)}(x) \lambda_{\alpha(k)} \partial_{\mu}$ 
a horizontal and holomorphic vector field generating the flow of diffeomorphisms. 
Together with the invariance $A \mapsto A + d_x\chi$ of the field equations, this 
geometrically realizes the gauge symmetries \eqref{HS SD gravity: gauge}.

\subsection{Twistor space}\label{subsection Twistor space}

The key point now is that the 2-form \eqref{F on correspondence space} is 
horizontal and therefore, since the spacetime is four dimensional, the field 
equation \eqref{Field equation for F in correspondence space} are equivalent to the 
requirement that $F$ must be simple, i.e., that there exits two horizontal 1-forms 
$E^{\dot{\alpha}}(x,\lambda) = {\rm d}x^{\mu} E^{\dot{\alpha}}_{\mu}(x,\lambda)$ with $\dot{\alpha} \in \dot{0}, \dot{1}$ such that
\begin{equation}
    F = 2\, E^{\dot{0}} \wedge E^{\dot{1}}\;.
\end{equation}
Note that, at this stage, there is no metric on our spacetime and the dotted spinor 
above is just a notation (ultimately it will be identified with dotted spinors of a 
self-dual spacetime). There is some ambiguity in the definition of 
$(E^{\dot{0}}, E^{\dot{1}})$ and, since $F$ is polynomial in $\lambda^{\alpha}$ 
one can always choose them such that
\begin{equation}
    E^{\dot{\alpha}}(x,\lambda) = {\rm d}x^{\mu} \left(\sum_{k\geq1}  
    E^{\dot{\alpha} \alpha(k)}_{\mu}(x) \lambda_{\alpha(k)} \right)\;.
\end{equation}
Making use of this frame one can now define the $2$-dimensional distribution 
$D_p \subset T_p \F$ at a point $p$ of the correspondence space as the kernel 
of the quadruplet 
$\left( {\rm d} \lambda^{0}, {\rm d} \lambda^{1}, E^{\dot{0}}, E^{\dot{1}} \right)$ :
\begin{equation}\label{Correspondence space: lax pair}
    D_p := \ker({\rm d} \lambda^{0})_p \cap \ker({\rm d} \lambda^{1 })_p \cap 
    \ker(E^{\dot{0}})_p \cap \ker(E^{\dot{1}})_p \subset T_p \F\;.
\end{equation}
It turns out that, as a result of the field equations, 
this distribution is integrable. 

\begin{proof}
We need to check that the following integrability constraints are satisfied 
$ {\rm d}({\rm d}\lambda^{0,1})\vert _{D} = 0\,$,
and ${\rm d}E^{\dot 0,\dot 1} \vert_{D} =0 $: 
On one hand we trivially have that
\begin{align}
	 {\rm d}({\rm d}\lambda^\alpha )
	= 0\;.
\end{align}
On the other hand, the two remaining integrability constraints are 
equivalent to
\begin{align}\label{eq equivalent integr cond}
	\forall X, Y \in D, \ {\rm d}F(X, Y, \cdot) = 0\;.
\end{align}
This is because, if $X, Y \in D$, the above reduces to
\begin{align}
	{\rm d}F(X, Y, \cdot) = 2 \,{\rm d}E^{\dot 0}(X,Y) \; E^{\dot 1}(\cdot) 
    - 2\, E^{\dot 0}(\cdot) \; {\rm d}E^{\dot 1}(X,Y)\;,
\end{align}
which, since $E^{\dot 0}$, $E^{\dot 1}$ are linearly independent, is equivalent 
to ${\rm d}E^{\dot 0}(X,Y) = 0\,$, $ {\rm d}E^{\dot 1}(X,Y) = 0\,$.

To prove that integrability follows from the field equations 
\eqref{Field equation for F in correspondence space}, it remains to show 
that \eqref{eq equivalent integr cond} does: On the one hand, we use that 
$F = {\rm d}_xA$ to rewrite \eqref{eq equivalent integr cond} as
\begin{equation}\label{eq equivalent integr cond2}
    \forall X, Y \in D, \;\;\frac{\partial F}{\partial \lambda^{\alpha}}(X, Y) 
    \,{\rm d}\lambda^{\alpha} = 0\;.
\end{equation}
On the other hand, differentiating the fields equations \eqref{Field equation for F in correspondence space}, one has
\begin{align}\label{d F wedge F}
	0 &=   {\rm d}{F}\,\wedge {F}= {\rm d}_{\lambda}F \,\wedge {F} 
    = 2\, \frac{\partial F}{\partial \lambda^{\alpha}} \wedge 
    {\rm d} \lambda^{\alpha}\wedge  E^{\dot 0} \wedge  E^{\dot 1}\;,
\end{align}
and contracting with $X,Y \in D$ gives \eqref{eq equivalent integr cond2} 
as claimed.
\end{proof}

As a result of the integrability of the distribution $D$, 
the 6 dimensional correspondence space $\F$ is foliated by 2 dimensional leaves. 
The twistor space $\T$ of the HS SD gravity solution is then obtained as the 
quotient by these leaves, see the resulting double fibration in Figure 
\ref{Fig: Double fibration}.

\begin{figure}[h]
    \centering
\begin{tikzcd}
     & \F = {\M} \times \mathbb{C}_0^2  \arrow[dl, "\pi_1"] \arrow[dr, "\pi_3"] & \\
	  {\M}   &   & \T \;.
\end{tikzcd}
   \caption{Double fibration of the correspondence space $\F$: 
   as a trivial bundle over spacetime $\M$ and as a space of leaves. } 
   \label{Fig: Double fibration}
\end{figure}
Note that the projection $\pi_2 : \F \to \IC^2_0$ on the correspondence space 
induces on the resulting twistor space $\T$ a natural holomorphic fibration
\begin{equation}\label{Twistor fibration}
    \pi : \T \longrightarrow \IC^2_0\;.
\end{equation}
This follows from the fact that the distribution $D$ is transverse to the 
vertical distribution $V :=\textrm{span}\left(\frac{\partial}{\partial \lambda^{0}}, \frac{\partial}{\partial \lambda^{1}} \right)\,$.

\paragraph{Action of gauge symmetries on twistor space}

As usual in twistor theory, the double fibration \ref{Fig: Double fibration} 
means that for each spacetime point $x^{\mu} \in \M$ one has a corresponding 
embedding 
$X_x : \mathbb{C}_0^2 \hookrightarrow \T\,$, obtained by considering 
$\pi_3 \circ \pi_1^{-1}(x^\mu)\,$. In this way, we obtain a 4-parameter family 
$X_x$ of holomorphic sections of the twistor fibration \eqref{Twistor fibration}. 

Now, a new phenomenon here is that the action of the symmetry 
\eqref{Diffeo in correspondence space} will give rise to a \emph{different}, 
alternative, four parameter family 
$X_{\phi(x)} : \mathbb{C}_0^2 \hookrightarrow \T\,$, by considering 
\begin{align}\label{eq action of the gauge symmetry on the curves}
	\pi_3 \circ \Phi \circ \pi_1^{-1}(x^\mu)\;,
\end{align}
as pictured in Figure \ref{fig: diffeo}. It might be highlighting to consider this phenomenon on the flat model.

\begin{figure}
    \centering
    \includegraphics[width=4in]{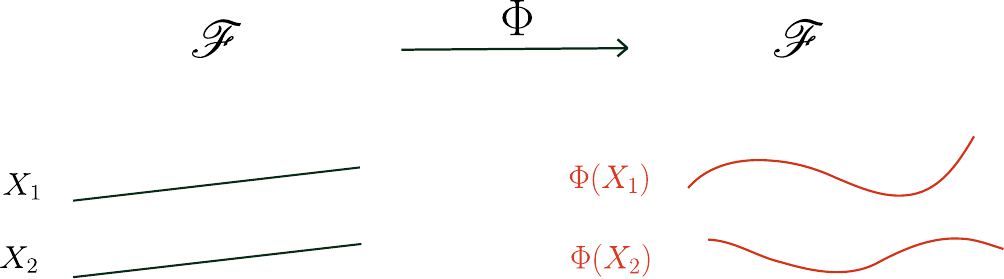}
    \caption{The symmetry $\Phi$ is acting as a diffeomorphism of $\F$. On the left are represented two points of the four-parameter family $\pi_1^{-1}(x^\mu)$ and on the right the same points, after the action of the symmetry $\Phi\circ \pi_1^{-1}(x^\mu)$.}
    \label{fig: diffeo}
\end{figure}

\paragraph{The flat model}

In this case the spacetime is just flat Minkowski space $\mathbb{M}^4$ and the twistor space is
$\mathbb{T}=\mathbb{C}^4_0 \setminus \mathbb{C}^2_0$ with usual complex structure.
Flat spinor indices can therefore be simply converted as spacetime vector indices : $x^\mu = x^{\alpha\dot\alpha}$. The incidence relation 
relates points $x^{\alpha\dot \alpha}$ of $\mathbb{M}^4$ with curves in twistor space and gives a 4-parameter family $X_x: \mathbb{C}^2_0 \hookrightarrow  \mathbb{T}$ of holomorphic embeddings
\begin{align}
X_x\left|
\begin{array}{ccccc}
\mathbb{C}^2_0 & \quad \longrightarrow &  \mathbb{C}^4_0 \setminus \mathbb{C}^2_0 \\[0.4em]
\lambda_\alpha & \quad \mapsto &
	\begin{pmatrix}
	x^{\dot\alpha\alpha}\lambda_{\alpha}	\\ \lambda_\alpha
	\end{pmatrix}.
\end{array}	\right.
\end{align}

Under a diffeomorphism of the type \eqref{Diffeo in correspondence space}, this relation get transformed.
In the correspondance space, it acts on the points as 
\begin{align}\label{Diffeo in flat space}
	\Phi \left|\begin{array}{ccc}
	    (x^\mu, \lambda_\alpha) &\longmapsto& (\varphi^{\mu}(x, \lambda_\alpha), \lambda_\alpha),
	\end{array}
	\right.
\end{align}
with\footnote{In principle the first term in the expansion might be a diffeomorphism; we drop this contribution which does not bring anything new.} 
\begin{align}\label{Flat model: correspondence space diffeo}
	\varphi^{\mu}(x, \lambda) 
	& = x^\mu + \sum_{k\geq 1} \varphi^{\mu\alpha(k)}(x)\lambda_{\alpha(k)}
\end{align}
Because the diffeomorphism is horizontal it will preserve the Lax pairs, and hence, the twistor space.
The projection to flat twistor space, given 
here by the incidence relation, implies that the initial 4-parameter family of holomorphic embeddings $X_x: \mathbb{C}^2_0 \hookrightarrow  \mathbb{C}^4_0$ is transformed into
\begin{align}
X_{\phi(x)}\left|
\begin{array}{ccccc}
\mathbb{C}^2_0 & \quad \longrightarrow &  \mathbb{C}^4_0 \\[0.2em]
\lambda_\alpha & \quad \mapsto &
	\begin{pmatrix}
	\varphi_x^{\alpha\dot\alpha}\lambda_\alpha\\ \lambda_\alpha
	\end{pmatrix}.
\end{array}	\right.
\end{align}
with 
\begin{align}
	\varphi_x^{\alpha\dot\alpha}\lambda_\alpha = x^{\alpha\dot\alpha}\lambda_\alpha + \sum_{k\geq0}
	\varphi^{\dot\alpha\alpha(k+1)}(x)\lambda_{\alpha(k+1)}.
\end{align}
Note in particular that the subspace of diffeomorphisms of the correspondence space \eqref{Flat model: correspondence space diffeo} such that $\varphi^{\dot\alpha\alpha, \alpha(k)} =0$ (with symmetrization implied on the $\alpha$ indices) acts trivially in twistor space.

\section{Nonlinear higher-spin graviton theorem}\label{section: Nonlinear HS graviton theorem}

\subsection{Curved twistor space}

Let $(\mu^{\dot \alpha} , \lambda_{\alpha})$ be holomorphic coordinates 
on the flat twistor space and let 
$\left\{ - , -\right\} = \epsilon^{\dot \alpha \dot \beta} \frac{\partial}{\partial \mu^{\dot \alpha}}\frac{\partial}{\partial \mu^{\dot \beta}}$ 
be the flat holomorphic Poisson structure.

Motivated by the previous section and the Euclidean results from 
\cite{Herfray:2022prf,Mason:2025pbz} we consider a deformation $\T$ 
of the complex structure on the non projective twistor space $\mathbb{C}^4_0\setminus \mathbb{C}^2_0\,$. 
In practice, we will work on the following open set
\begin{equation}\label{Boundness domain}
    \mathbb{T} = \left\{\; (\mu^{\dot \alpha}, \lambda_{\alpha} ) \in \mathbb{C}^4 \; \vert\; \big|\mu^{\dot \alpha}\big| < R \big|\lambda^{\alpha}\big| \;, \;  \big|\lambda^{\alpha}\big| < C \;\right\}\;.
\end{equation}
The first condition, with a fixed real constant $R$, is needed to avoid 
working too close to the plane $\lambda=0\,$, where singularities of the 
complex structure are expected\footnote{From the spacetime perspective, 
this amounts to restricting to  $|x|<R$ so as to provide a local 
construction on spacetime avoiding potential singularities as $x\rightarrow \infty$. 
See Appendix \ref{Apdx: Euclidean fields} for a discussion in Euclidean signature.}. 
The second condition, with a real constant $0<C<1$ that will be adapted along the 
article, is introduced in order to  keep control of the convergence of the different 
series at play. In what follows, in order to lighten the text, we will sometime keep 
implicit the condition $0< |\lambda| < C$ and refer abusively to this domain as 
$\mathbb{C}^2_0\,$.\\

Thus we will consider deformations of the complex structure of $\mathbb{T}$ such that
\begin{enumerate}[(i)]
    \item the holomorphic projection $\mathbb{T} \to \mathbb{C}^2_0\,$,
    \label{item holom projection}
    \item the holomorphic Poisson structure $\left\{ - , -\right\}$ is preserved,
    \label{item holom poisson}
    \item the Euler vector field $E=\mu^{\dot{\alpha}} \frac{\partial}{\partial \mu^{\dot{\alpha}}} + \lambda_{\alpha} \frac{\partial}{\partial \lambda_{\alpha}}$ is $(1,0)\,$, but \emph{is not necessarily holomorphic}.
    \label{item euler not holom}
\end{enumerate}
Crucially, as compared to the classical nonlinear graviton theorem 
\cite{Penrose:1976js} (see in particular \cite{eastwood_tod_1982} for a 
presentation closely related to ours), we relax the holomorphicity of the 
Euler vector field. This implies that these deformations can not be identified 
with deformations of the projective twistor space 
$\mathbb{PT} = \mathbb{C}\textrm{P}^3 \setminus \mathbb{C}\textrm{P}^1\,$.

Assumption \ref{item holom poisson} imply that the 
deformations are locally given by perturbations of the Dolbeault operator 
of the form  $\bar{\nabla} = \bar{\partial} + \{\h, - \}\,$ where initially $\h$ is an arbitrary $(0,1)$-form on $\mathbb{T}$.
Integrability is then given by \begin{equation}\label{eq integrability}
    \bar{\partial} \h + \{ \h,\h\}=0\;.
\end{equation}
We can further restrict the form of $\h$ as follows. It is standard result in twistor theory that the requirements 
\ref{item holom projection} and \ref{item holom poisson} imply that the 
deformations are locally given by $\bar{\nabla} = \bar{\partial} + \{\h, - \}\,$, 
where $\h = h_{\dot\alpha}\, {\rm d}{\bar{\lambda}}^{\dot{\alpha}} \in \Omega^{(0,1)} (\mathbb{T})$ is a horizontal $(0,1)$-form. 
This is because, at fixed $\lambda$, the fibres of the projection 
to $\mathbb{C}^2_0$ are Stein (being convex subsets of $\mathbb{C}^2$ as 
follows from \eqref{Boundness domain}).  Thus firstly we can find holomorphic coordinates $\mu^{\dot \alpha}$ that can then also be further restricted  to put the Poisson structure into its normal form by Darboux's theorem.  This can be done smoothly  in $\lambda$ but not necessarily holomorphically.  

After this, the almost complex structure is given by
\begin{align}
    T^{0,1}\T = Span\left( \frac{\partial}{\partial \bar\mu^{\alpha}} , \frac{\partial}
    {\partial \bar{\lambda}^{\dot\alpha}} + \left\{ h_{\dot\alpha}, - \right\} \right)\;,
\end{align}
Requiring \ref{item euler not holom}, i.e. that the Euler vector field 
is $(1,0)$, together with \eqref{eq integrability}
now restricts to deformations given by $\h$ of the form
\begin{align}\label{Twistor space deformations: deformation we consider}
   \h = h \, \bar{\lambda}^{\dot\alpha}{\rm d}\bar{\lambda}_{\dot\alpha} &\in \Omega^{(0,1)} (\mathbb{T})\;, & \text{ with } && \bar{\partial}\h &=0\;.
\end{align}
We have now completely solved the integrability condition \eqref{eq integrability} and all the information of the deformation is encoded in $\h$.\\

The integrability condition is now equivalent to the 
holomorphicity of $\h$ i.e.
\begin{align}\label{Twistor space deformations: nice h properties}
    \frac{\partial}{\partial \bar{\mu}^{\alpha}}h &=0\;, 
    & \bar{\lambda}_{\dot\alpha} \frac{\partial}{\partial \bar{\lambda}_{\dot\alpha}}h = -2 h\;.
\end{align}
As a result $\h$ must be homogeneous of degree $0$ in the anti-holomorphic coordinates:
\begin{align}
    \mathcal{L}_{\bar{E}} \h &=0,& \mathcal{L}_{E} \h &\neq 2 \h\;.
\end{align}
The second equation here is meant to emphasize the difference with the standard nonlinear graviton construction:

In the classical nonlinear graviton theorem \cite{Penrose:1976js}, 
one would require that the Euler vector field $E$ is \emph{holomorphic}, 
and not merely $(1,0)$, which is equivalent to $\mathcal{L}_{E} \h = 2 \h\,$. 
As a result, in the classical construction, $\h$ defines a section of 
$\Omega^{(0,1)}(\mathbb{PT},\mathcal{O}(2))$ on the projective twistor space 
$\mathbb{PT}$ and the corresponding complex structure is interpreted as a 
deformation $\PT$ of the complex structure on $\mathbb{PT}\,$. 
However, in the present article, the weaker requirement 
\ref{item euler not holom}  means that there will be no constraint 
on the homogeneity of $\h$ in the holomorphic coordinates. 
Therefore, as opposed to the classical twistor construction, 
$\h$ does not in general define a section of 
$\Omega^{(0,1)}(\mathbb{PT}, \mathcal{O}(2))\,$ and, as a result, 
generically does not parametrize a deformation of projective twistor 
space $\mathbb{PT}\,$.\\

The fact that, under the present assumption, $\h$ has no fixed homogeneity also means  that it cannot be interpreted, via the Penrose transform, as a spin-$2$ field. Therefore we loose the interpretation that an infinitesimal deformation of the complex structure is generated by a graviton. Nevertheless, the deformations that we consider are still rigid enough to be interpreted as being generated by a tower of higher spin fields.

\begin{Proposition}\label{Prop decomposition of deformations}
	 For a deformation of the twistor space $\bar{\nabla} = \bar{\partial} + \{\h, - \}$ given by \eqref{Twistor space deformations: deformation we consider}, the $(0,1)$-form $\h$ can always be 
	written, on  \eqref{Boundness domain}, as 
	\begin{align}\label{Decomposition of h in homogenoeity}
		\h(\mu^{\dot\alpha}, \lambda_\alpha)&= 
		\sum_{k\in \frac{1}{2}\mathbb{Z}} \h^{(k)}\;, 
	\end{align}
    where $\h^{(k)} = h^{(k)}(\mu , \lambda)\; \bar{\lambda}^{\dot\alpha}
    {\rm d}\bar{\lambda}_{\dot\alpha}$ 
    are holomorphic $(0,1)$-forms on $\mathbb{T}$ of homogeneity $(2k-2,0)\,$. 
    They thus descend to $\mathcal{O}(2k-2)$-valued holomorphic forms on 
    projective twistor space $\mathbb{PT}$ and, by Penrose transform, 
    correspond to linearized spin-$k$ fields propagating on Minkowski space. 
\end{Proposition}

Although this can be proved quite directly,\footnote{Here we can trivialize the line bundle $\mathbb{C}^4\rightarrow \mathbb{C}\textrm{P}^3$ over the complements of the coordinate planes so that the bundle given by $Z\rightarrow e^{i\theta}Z$ becomes trivial on each open set and we can simply perform Fourier expansion on each such set, but then each $h^{(k)}$ will patch with the appropriate weight on each overlap making it a section of $\mathcal{O}(n)$ as desired. } as a result of the proof, we will see that $h^{(k)}(\mu , \lambda)$ has the following expansion that will be needed later:
	\begin{align}\label{Proposition: homogeneity decomposition equation}
		h^{(k)}(\mu^{\dot\alpha}, \lambda_\alpha)\; 
		&= \sum_{n\in \mathbb{N}}\sum_{j\in \frac{1}{2}\mathbb{N}} h_{(l)\;\dot\beta(n)}^{\;\alpha(2j)}\mu^{\dot\beta(n)}
		\frac{\lambda_{\alpha(j+l)}\hat\lambda_{\alpha(j-l)}}{\langle\lambda\hat\lambda^{}\rangle^{j-l+2}},&\text{s.t. } &\begin{cases}
		   l = k -\frac{n}{2}\\
          j  \geq |l|
		\end{cases}\;,
	\end{align}
where $\langle \lambda \hat\lambda \rangle = \lambda_{\alpha}\hat{\lambda}^{\alpha} = \lambda_{\alpha}h^{\alpha \dot \alpha }\bar{\lambda}_{\alpha}$ 
is a choice\footnote{As a result, the decomposition \eqref{Proposition: homogeneity decomposition equation}, which depends on this choice, breaks $SL(2,\mathbb{C})$ 
invariance, note however that the resulting homogeneity decomposition \eqref{Decomposition of h in homogenoeity} is $SL(2,\mathbb{C})$ invariant and thus independent of this choice.} of positive definite inner product on $\mathbb{C}^2$. \\

\begin{proof}
The proof start by making use of 
\eqref{Twistor space deformations: nice h properties} to write, 
in a neighbourhood of $\mu=0$ and therefore adapting the constant 
$C$ in \eqref{Boundness domain},
	\begin{align}
		h(\mu^{\dot\alpha}, \lambda_\alpha)\; 
		= \sum_{n \in \mathbb{N}} h_{\dot\beta(n)}(\lambda,\bar\lambda)\,\mu^{\dot\beta(n)}\;,
	\end{align}
where $h_{\dot\beta(n)}(\lambda,\bar\lambda)$ is of homogeneous degree $-2$ 
in $\bar\lambda_{\dot\alpha}\,$. 
The decomposition \eqref{Proposition: homogeneity decomposition 
equation} is then obtained by making use of 
Lemma \ref{lemma decomposition on C^2} below.
\end{proof}

\begin{Lemma}\label{lemma decomposition on C^2}
	Every  smooth function $f(\lambda_\alpha,\bar\lambda_{\dot\alpha})$ on $\IC^2_0$ which is of 
	homogeneity $\bar{h}$ in $\bar\lambda_{\dot\alpha}$ can 
	be written as
	\begin{align}\label{Lemma: decomposition in spin weight (equation)}
		f(\lambda_\alpha,\bar\lambda_{\dot\alpha}) =
		\sum_{s\in \frac{1}{2}\mathbb{Z}}\; \left(\sum_{j\geq \vert s\vert}f^{\alpha(2j)}_{(s)}\frac{\lambda_{\alpha(j-s)}\hat\lambda_{\alpha(j+s)}}{\langle\lambda\hat\lambda\rangle^{j+s-\bar{h}}} \right)
	\end{align}
	where $f^{\alpha(2j)}_{(s)}$ are complex numbers. In particular each term 
    in parenthesis define a section of $\mathcal{O}(2s+\bar{h}, \bar h)$.
\end{Lemma}

\begin{proof}
	On $S^3$ with complex coordinates $(z_1, z_2)$ the $SU(2)$ invariant inner 
    product is such  that 
	$\vert z_1\vert^2 + \vert z_2\vert^2= r_0$.
	We know from Peter Weyl theorem that $L^2$ functions on $S^3$ can be decomposed 
    in irreducible representations labelled by a non negative half integer $j$ with 
    coefficient the spherical harmonics functions of degree $j$ : 
	\begin{align}\label{eq peter weyl}
		f(r_0,\theta_1, \theta_2, \theta_3) = \sum_{j \in \frac{1}{2}\mathbb{N}}\sum_{m=-j}^j\sum_{s=-j}^j f^{j,m,n}(r_0)\;  {}_s Y_{j,m}(\theta_1, \theta_2, \theta_3)\;,
	\end{align} 
	where $(\theta_i)$ are angular spherical coordinates on $S^3\,$. 
    Now the 3d spherical harmonics ${}_s Y_{j,m}(\theta_1, \theta_2, \theta_3)$ 
    are related to the 2d \emph{spin-weighted} spherical harmonics 
    ${}_sY_{j,m}(\theta_1, \theta_2)$ as ${}_s Y_{j,m}(\theta_1, \theta_2, \theta_3) =  {}_sY_{j,m}(\theta_1, \theta_2)e^{is\theta_3}\,$. 
    The decomposition \eqref{eq peter weyl} can thus equivalently be seen as 
    a decomposition in spin-weighted functions on $S^2\,$,
\begin{align}
	f(r_0,\theta_1, \theta_2, \theta_3) = \sum_{s \in \frac{1}{2}\mathbb{Z}}f^{(s)}(r_0,\theta_1, \theta_2, \theta_3)\;,
\end{align}
where $f^{(s)} = F^{(s)}(r_0,\theta_1, \theta_2) e^{is\theta_3}$ is the 
homogeneous realizations of a spin-$s$ weighted functions:
\begin{align}
	F^{(s)}(r_0,\theta_1, \theta_2) = \sum_{j \geq \vert s \vert} 
    \sum_{m=-j}^{j} f^{j,m,s}(r_0)\; {}_sY_{j,m}(\theta_1, \theta_2)\;.
\end{align}
Making use of the standard coordinates 
$(\lambda_\alpha, \bar\lambda_{\dot\alpha})$ 
on $\IC^2_0$ one obtains
\begin{align}
	f^{(s)}(\lambda_\alpha, \bar\lambda_{\dot\alpha}) = \sum_{j \geq \vert s \vert}  f^{\alpha(2j)}_{(s)}(r_0)\; \frac{\lambda_{\alpha(j+s)}\hat\lambda_{\alpha(j-s)}}{\langle\lambda\hat\lambda\rangle^{j}}\;,
\end{align}
where we used \cite{eastwood_tod_1982} to express the spherical harmonics in 
homogeneous notations. Putting everything together, and since 
$r_0=\langle \lambda \hat{\lambda}\rangle$, one obtains a 
decomposition of functions on $\IC^2_0 = S^3 \times \IR_{>0}$ as
\begin{align}\label{eq decomposition with radius}
	f(\lambda_\alpha, \bar\lambda_{\dot\alpha}) = \sum_{s \in \frac{1}{2}\mathbb{Z}}\; \sum_{j \geq \vert s \vert}  \tilde{f}_{(s)}^{\alpha(2j)}(r_0)\; \lambda_{\alpha(j+s)}\hat\lambda_{\alpha(j-s)}\;.
\end{align}

Now, by hypothesis, $\tilde{f}(\lambda, \bar\lambda)$ is homogeneous degree  $\bar{h}$ in $\bar\lambda_{\dot\alpha}$, therefore $\tilde{f}_{(s)}^{\alpha(2j)}(r_0 = \lambda\hat\lambda)$ is of homogeneity $(\bar{h}-j+s)\,$. As a result 
 $\tilde{f}_{(s)}^{\alpha(2j)}(r_0) = f_{(s)}^{\alpha(2j)} \langle\lambda\hat\lambda \rangle^{\bar{h}-j+s}$ and the decomposition \eqref{eq decomposition with radius} can then be rewritten as  \eqref{Lemma: decomposition in spin weight (equation)}.
\end{proof}

At this stage, one sees from \eqref{Decomposition of h in homogenoeity} 
that the deformations are sourced by both positive and negative helicity 
higher-spin fields. Deformations which are only generated by positive 
helicity fields can however be characterized by their good behaviour 
at the origin $(\mu^{\dot\alpha}, \lambda_{\alpha}) =(0,0)$ of the twistor space : 
\begin{Proposition}\label{prop h starts at 2}
The deformation $\bar{\nabla} = \bar{\partial} + \{\h, - \}$ satisfying 
\eqref{Twistor space deformations: deformation we consider} 
is bounded \footnote{Singularities necessarily arise along 
$\left(\mu\neq 0,\lambda=0\right)$ in order to have a non trivial deformation. 
One therefore cannot expect the deformation to be bounded in a neighbourhood 
of the origin. See Appendix \ref{Apdx: Euclidean fields} for a discussion in 
the context of Euclidean signature.} on the twistor space
\eqref{Boundness domain} if and only if the decomposition of 
Proposition \ref{Prop decomposition of deformations} 
is such that $k \geq 2\,$, i.e., if they are only generated by higher 
spin fields of positive helicity $s\geq2\,$.
\end{Proposition}

\begin{proof}
    From \eqref{Decomposition of h in homogenoeity} and 
    $\bar{\nabla} = \bar{\partial} + \{\h, - \}\,$, one has
    \begin{equation}
        \bar{\nabla} = {\rm d}\bar{\lambda}_{\dot\alpha}
        \frac{\partial }{\partial \bar{\lambda}_{\dot\alpha} } 
        + {\rm d}\bar{\mu}^{\alpha}\frac{\partial }{\partial \bar{\mu}^{\alpha} }
        +  \sum_{k\in \frac{1}{2}\mathbb{Z}} \left( 
        \bar{\lambda}^{\dot\alpha}d\bar{\lambda}_{\dot\alpha} \,
        \epsilon^{\dot\beta\dot\gamma}\,
        \frac{\partial h^{(k)}}{\partial \mu^{\dot \beta} }  
        \frac{\partial }{\partial \mu^{\dot \gamma} }\right)\;,
    \end{equation}
where each term in the sum has homogeneity $k-2$ in 
$(\mu^{\dot\alpha}, \lambda_{\alpha})$ and zero in 
$(\bar{\mu}^{\alpha}, \bar{\lambda}_{\dot\alpha})\,$. 
Therefore $\bar{\nabla}\big|_{(t\mu_0,t\lambda_0)}$ is bounded as $t$ 
goes to zero if and only if $h^{(k)}=0$ for any $k<2\,$. As a result the 
deformed complex structure can only be bounded on \eqref{Boundness domain} 
if $k\geq2\,$. Finally, the requirement that the fields are bounded on 
\eqref{Boundness domain} is not a constraint when the spin is higher or equal 
to two: one way to see this is to consider the Euclidean realization of the 
fields as in Appendix \ref{Apdx: Euclidean fields}.
\end{proof}

The deformations satisfying the requirement 
\ref{item holom projection}, \ref{item holom poisson} 
and \ref{item euler not holom} naturally incorporate both negative and 
positive helicity higher-spin fields. As such, and in view of the results of 
\cite{Mason:2025pbz}, this is a very promising road to incorporate the 
interactions of chiral HiSGra. We shall come back on this idea by the 
end of the present article. In the present work we will rather try to 
connect with higher-spin self-dual gravity, whose field equations 
only involve positive higher-spin fields. 
As a result, and in order to prevent the interaction of 
negative helicity fields, 
in what follows we will require as a final constraint on the deformations:

\begin{enumerate}[(i),resume]
    \item\label{item bounded deformation} the deformation $\bar{\nabla}$ of the complex structure must remain bounded on the twistor space \eqref{Boundness domain}. 
\end{enumerate}

As we shall now see, this assumption greatly simplifies the analysis of the moduli space of holomorphic planes.

\subsection{Holomorphic planes in twistor space}

\subsubsection{The higher-spin space $\HST$}

By analogy with the case of self-dual gravity with zero cosmological constant, we will define the curved higher spin spacetime $\HST$ as the space of holomorphic sections $\Gamma_{1,0}\left( \T \right)$ of
\begin{equation}
    \T \to \IC^2_0\;.
\end{equation}
Such sections are given by
\begin{equation}\label{Sections of T}
   X: \left| \lambda_{\alpha} \mapsto \begin{pmatrix} F^{\dot \alpha}(\lambda, \bar{\lambda}) \\ \lambda_{\alpha} \end{pmatrix} \right.
\end{equation}
where $F$ is not required to have any homogeneity property but satisfies
\begin{equation}\label{eq holomorphicity of curves}
   \frac{\partial}{\partial \bar{\lambda}^{\dot{\beta}}} F^{\dot \alpha} 
   -  \bar{\lambda}_{\dot{\beta}}\{ h , \mu^{\dot \alpha} \} =0\;.
\end{equation}
In particular we expect the higher spin spacetime $\HST$ to be infinite 
dimensional (see the flat model below). Now the situation is not as 
straightforward as might first appear as projectively,
\begin{equation}
    \begin{alignedat}{2}
    H^{(0,1)}\left( \mathbb{C}P^1 , \mathcal{O}(2s-2) \right)&\neq0& \qquad\text{for}&\quad s\leq 0\;,\\
    H^{(0,1)}\left( \mathbb{C}P^1 , \mathcal{O}(2s-2) \right)&=0& \qquad\text{for}&\quad s>0\;,
\end{alignedat}
\end{equation}
and anticipating the detailed analysis of the solutions which 
will be provided below, one should not expect that there is any solutions 
to \eqref{eq holomorphicity of curves} if negative helicity fields are sourcing 
the deformation of the complex structure\footnote{As a result, whenever negative 
helicity fields are sourcing the deformations, one is forced to to allow for 
singularity in the embedding \eqref{Sections of T}, we will come back on 
this in the final section of this article.}. 
For this reason, we will from now on suppose that our deformations satisfy the 
requirement \ref{item bounded deformation}, and are only sourced by fields 
of helicity greater than or equal to $2$ (see Proposition \ref{prop h starts at 2}).

The deformation $\h$ is not expected to be regular on the plane 
$\lambda_\alpha = (0,0)$, however the requirement \ref{item bounded deformation} 
means that the origin $(\mu^{\dot\alpha}, \lambda_{\alpha})=(0,0)$ is
well behaved as long as we stay in the 
domain \eqref{Boundness domain}. 
As a result, it is natural to define the higher spins spacetime $\HST$ 
as the space of sections \eqref{Sections of T} passing through the twistor 
space \eqref{Boundness domain}, i.e. such that, 
near $\lambda=0$,
\begin{equation}\label{holomorphic curves: 0 condition}
    \left|F^{\dot\alpha}(\lambda,\bar{\lambda}) \right| < R|\lambda^{\alpha}|\;.
\end{equation}
Such sections are depicted in Figure \ref{fig: section of T}.\\
\begin{figure}
    \centering
    \includegraphics[width=3in]{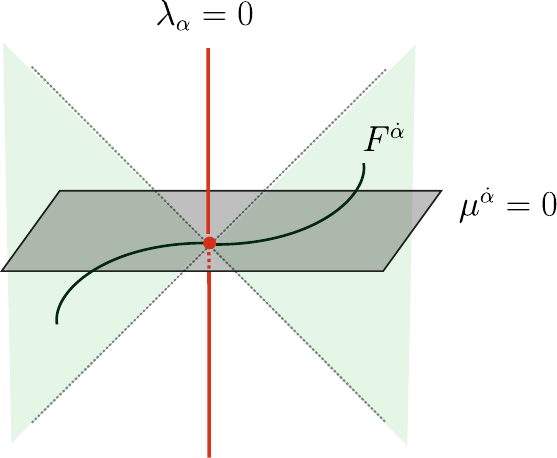}
    \caption{The curved twistor space might be expected to be singular in some neighbouhood $\lambda_\alpha = (0,0)$, depicted in red. The holomorphic planes are required to stay within the green region of the twistor space which will be better behaved.}
    \label{fig: section of T}
\end{figure}

From equation \eqref{eq holomorphicity of curves} one can deduce that 
the curves are of homogeneity zero in $\bar\lambda_{\dot\alpha}\,$. 
This fact, together with
the condition \eqref{holomorphic curves: 0 condition}, 
implies the following.
\begin{Proposition}\label{prop F starts at 1}
    If $F$ is a smooth function on $\IC^2_0$ of homogeneity zero in $\bar\lambda_{\dot\alpha}$  such that, for some arbitrary constant $K$ and in some neighbourhood of the origin, $\left|F^{\dot\alpha}(\lambda,\bar{\lambda}) \right| < K |\lambda_{\alpha}|$, then it can be expanded as 
    \begin{align}
        F(\lambda_\alpha,\bar\lambda_{\dot\alpha}) = 
        \sum_{\overset{l \geqslant 1}{l \in \mathbb{Z}}} F_{(l)}(\lambda_\alpha,\bar\lambda_{\dot\alpha})\;,
    \end{align}
    where $F_{(l)}(\lambda_\alpha,\bar\lambda_{\dot\alpha})$ is a section of $\mathcal{O}(l)\,$.
\end{Proposition}
\begin{proof}
    The expansion on $l\in \mathbb{Z}$ comes from Lemma \ref{lemma decomposition on C^2}. It remains to show that the sum is only on $l\geq 1\,$.
    On the one hand
    \begin{align}
        \lim_{t\rightarrow 0} \left|F(t\lambda_\alpha, \bar t\bar\lambda_{\dot\alpha})\right| < K\lim_{t\rightarrow 0} \left(t |\lambda_{\alpha}|\right) =0\;,
    \end{align}
    on the other hand, for each homogeneity $l$ term, 
    \begin{align}
        \lim_{t\rightarrow 0} F_{(l)}(t\lambda_\alpha, \bar t\bar\lambda_{\dot\alpha}) =\lim_{t\rightarrow 0} t^l F_{(l)}(\lambda_\alpha, \bar \lambda_{\dot\alpha})\;, 
    \end{align}
    which is zero if and only if $l\geq 1\,$.
    Therefore, in the sum, introducing the terms with 
    $l <1$ will lead to nonzero terms, possibly divergent, with no possible cancellation because of the form of $F_{(l)}\,$.
\end{proof}

Before further investigating the higher-spin space $\HST\,$, 
let us consider the flat model.

\paragraph{The flat model}

For the flat twistor space, holomorphic sections \eqref{Sections of T} 
satisfying \eqref{holomorphic curves: 0 condition} would be given 
by\footnote{The function $F(\lambda,\bar{\lambda})$ is here only assumed to 
holomorphic on $\mathbb{C}^2_0$ but by Hartog's theorem it must in fact uniquely 
extend to a holomorphic function on $\mathbb{C}^2$, resulting in the existence 
of the Taylor expansion.}
\begin{equation}\label{Flat model: holomorphic planes}
    F^{\dot \alpha}(\lambda, \bar{\lambda})  
    =  \sum_{s\geq 1} X^{\dot \alpha \alpha(s)}\lambda_{\alpha(s)}
\end{equation}
where $X = \left( X^{\dot \alpha \alpha}, X^{\dot \alpha \alpha(2)}, ...  \right)$ 
are coordinates on the higher-spin spacetime $\HST$. This infinite dimensional space 
is an affine spacetime isomorphic to
\begin{equation}
    \HST \simeq \frac{SL(2,\mathbb{C}) \ltimes 
    \left( \mathcal{O}_0(\mathbb{C}^2) \oplus  \mathcal{O}_0(\mathbb{C}^2) 
    \right)}{SL(2,\mathbb{C}) }\;,
\end{equation}
where we denote by $\mathcal{O}_0(\mathbb{C}^2)$ the space of holomorphic functions 
on $\mathbb{C}^2$ vanishing at the origin. The condition 
\eqref{holomorphic curves: 0 condition} then amounts to restricting 
$X^{\dot \alpha \alpha}$ to a finite ball around the origin.
Note that there is here a canonical embedding of the Poincaré group inside 
this group corresponding to the fact that usual spacetime, parametrized by 
$ X^{\dot \alpha \alpha}$, is canonically embedded inside the higher spacetime 
$\HST$ by restricting to homogeneous degree 1 sections.

\subsubsection{Holomorphic planes in twistor space}

Let us now investigate further the properties of the higher-spin space $\HST$, 
defined as the moduli space of holomorphic embeddings \eqref{Sections of T} 
satisfying \eqref{holomorphic curves: 0 condition}, under the assumption 
\ref{item bounded deformation} on the complex structure.\\

To do so, we start from equation \eqref{eq holomorphicity of curves}, 
which can be rewritten as
\begin{align}\label{eq holomorphicity of curves 2}
	\bar{\partial} F^{\dot\beta} = \frac{\partial \h}{\partial \mu_{\dot\beta}}\Big|_{\mu^{\dot\alpha} = F^{\dot\alpha}}\;, 
\end{align}
and decompose it, on the domain \eqref{Boundness domain}, 
in homogeneity degrees of $\lambda_\alpha$ by making use of Propositions 
\ref{Prop decomposition of deformations}, \ref{prop h starts at 2}, 
\ref{prop F starts at 1}, and Lemma \ref{lemma decomposition on C^2}.  
Namely, from Proposition \ref{prop F starts at 1}, one has
\begin{align}\label{Formal solution of curves}
	F^{\dot\beta} = \sum_{l \geq 1} F^{\dot\beta}_{(l)}\;, \quad \text{with } l\in \mathbb{Z} \text{ and } F^{\dot\beta}_{(l)} \text{ section of } 
    \mathcal{O}(l)\;,
\end{align}
whereas, from Proposition \ref{Prop decomposition of deformations},
\begin{align}
	 h = \sum_{\overset{k\in\mathbb{Z}}{k\geq 2}}\sum_{n \in \mathbb{N}}{h}_{(k)}^{\dot\beta(n)}\mu_{\dot\beta(n)} \quad \text{with } {h}_{(k)}^{\dot\beta(n)} \text{ section of } \mathcal{O}(k-n,-2)\;.
\end{align}
The integer $k$ being here related to the physical helicity  $s$ 
via Penrose transform and $k = 2s-2$. The right-hand-side of 
\eqref{eq holomorphicity of curves 2} can then be expressed as
\begin{align}
	\sum_{k\geq 2}\sum_{n\geq 1} n\,{h}_{(k)}^{\dot\beta\dot\beta(n-1)}
    F_{\dot\beta(n-1)} 
	= \sum_{k\geq 2}\sum_{n\geq 1} \left( 
    n\,{h}_{(k)}^{\dot\beta\dot\beta(n-1)}\left(\prod_{i=1}^{n-1} 
    \sum_{m_i \geq 1} F_{\dot\beta_i(m_i)} \right)\right)\;.
\end{align}
At fixed $(n,k,m_i)$ the right-hand-side has homogeneity degree 
$k-n + \sum_{i=1}^{n-1} m_i$ in $\lambda_\alpha\,$. 
Therefore, the homogeneity $l$ terms of the right-hand-side of \eqref{eq holomorphicity of curves 2} are given by the combinations of 
$(n,k,m_i)$ such that 
\begin{align}\label{eq relation between homogeneities}
	l = k -1   +\sum_{i=1}^{n-1} (m_i-1)  \;.
\end{align}
Note that, because $\forall i\,, m_i \geq 1$, then for a fixed 
homogeneity $l$, the allowed $k$ are such that 
\begin{align}
	 2 \leq k \leq l+ 1\;.
\end{align}
In particular for $l=1$, the only possible value for $k$ 
is $k=2$ (and thus $m_i=1)$:
\begin{align}\label{holomorphic embeding: spin 2}
	\bar{\partial} F^{\dot\beta}_{(1)} = \frac{\partial \h_{(2)}}{\partial \mu_{\dot\beta}}\Big|_{\mu^{\dot\alpha} = F^{\dot\alpha}_{(1)}}\;.
\end{align}
From the usual nonlinear graviton theorem\footnote{Here and everywhere $\PT$ will denote the curved projective twistor space whose deformation is sourced by $\h_{(2)}\in\Omega^{(0,1)}(\mathbb{PT},\mathcal{O}(2))$ and $M$ the (four dimensional) self-dual spacetime associated to it by the standard nonlinear graviton theorem. This is as opposed to our curved twistor space $\T$ associated with the (infinite dimensional) HS space $\HST$.}, it admits a 4-parameter family of solutions $F_{(1)}(x): \mathbb{C}P^1 \hookrightarrow \PT$ with $\PT$ the projective curved  twistor space constructed from $\h_{(2)}\,$.

Furthermore, $k \geq 2$ means, together with \eqref{eq relation between homogeneities}, that $l+n -\sum m_i \geq 2\,$, and thus
\begin{align}\label{eq inequalities on m_i}
    \sum_{i=1}^{n-1}m_i \leq l+n-2  \qquad \Rightarrow \qquad \ 
   n-2 + m_{i^*} \;\leq\;  \sum_{i=1}^{n-2}m_i + m_{i^*} \;\leq\; l+n-2\;,
\end{align}
from which one deduces
\begin{align}
 \forall i^*,\    m_{i^*} \leq l\;.
\end{align}
This reproduces the triangular structure of higher-spin self-dual gravity 
equations: the HS fields sourcing the equation at order $l$ are all 
the fields of helicity $s$ such that $s \leq (l+3)/2\,$.

From these inequalities one can deduce that in the limit case 
where $k=l+1$, 
the only possible $m_i$ are the integers such that 
 $\forall i, m_i= 1\,$. 
 This will thus give, for the right-hand-side of \eqref{eq holomorphicity of curves 2} and at each order $l$, a term of the type
 \begin{align}
 	&\frac{\partial \h_{(l+1)}}{\partial \mu_{\dot\beta}}\Big|_{\mu^{\dot\beta} = F_{(1)}^{\dot\beta}}\;. 
 \end{align}

There is another important implication of \eqref{eq inequalities on m_i}. From the last inequality, one can see that if one of the $m_i$ is chosen equal to $l$, all the others are forced to be $1$, and therefore $k$ to be $2$. Altogether, this means that,
at fixed homogeneity $l$, the only term involving $F_{(l)}$ will bes
\begin{align}
    -\frac{\partial^2 \h_{(2)}}{\partial \mu_{\dot\beta}\partial \mu^{\dot\gamma}}\Big|_{\mu^{\dot\beta} = F_{(1)}^{\dot\beta}} F_{(l)}^{{\dot\gamma}}\;,
\end{align}
which is linear in $F_{(l)}\,$.
All the other terms appearing with homogeneity $l$
will involve factors $F_{(m)}$ with $m<l\,$.\\

Putting all the previous remarks together, one deduces that the structure of the homogeneity $l>1$ equation is
\vspace*{3pt}%
\begin{align}\label{holomorphic embeding: structure of the equation}
	\bar{\partial} F^{\dot\beta}_{(l)} + \frac{\partial^2 \h_{(2)}}{\partial \mu_{\dot\beta}\partial \mu^{\dot\gamma}}\Bigg|_{\mu = F_{(1)}} \hspace{-0.8cm}F_{(l)}^{{\dot\gamma}} = f^{\dot \beta}\Big(h_{(2)}, \dots, h_{(l+1)}, F_{(1)}, \dots, F_{(l-1)} \Big)\;.
\end{align}
\vspace{-0.2cm} \\

This equation has a direct interpretation\footnote{Here we will 
use the `hat' notation to signal geometrical objects related to the standard 
twistor space, e.g. $\hat{X} : \mathbb{C}\textrm{P}^1 \to \PT$ for a 
holomorphic curve and $\hat{N} \to Im \hat{X}$ for its normal bundle in $\PT$.} 
on $\PT$: The left hand side is the action of the Dolbeaut derivative, 
on a section of the weighted normal bundle 
$\hat{N} \times \mathcal{O}(l-1) = \mathcal{O}(l)\oplus \mathcal{O}(l)$ 
of the curve $F_{(1)}: \mathbb{C}P^1 \hookrightarrow \PT$.
This, together with the usual nonlinear graviton theorem,  
proves that equations \eqref{eq holomorphicity of curves}
can be solved recursively, order by order in the homogeneity degree. 
At each step, the solution is unique up to a solution of the homogeneous 
equation, a finite dimensional term which one can take to be parametrized by 
$x^{\underline{\dot{\alpha}}\alpha(l)}\,$. 
What is more, we saw that the lowest homogeneity degree, $l=1$,  of \eqref{eq holomorphicity of curves 2} corresponds to equations for the spin-$2$ field 
\eqref{holomorphic embeding: spin 2}, which by Penrose's nonlinear graviton 
theorem, possess a four dimensional space of solutions $\M$, 
here parametrized by $x^{\mu}=x^{\underline{\dot{\alpha}}\alpha}\,$. 
This means that there is a canonical projection $\HST \to \M$ from the 
(infinite dimensional) space of holomorphic planes $\HST$ to a usual four 
dimensional self-dual spacetime $\M\,$. Finally, given smallness 
assumptions on the $h_{(k)}$ and $F_{(k<l)}$, the right hand sides of 
\eqref{holomorphic embeding: structure of the equation} can be guaranteed 
to be small. This is because, by elliptic regularity, an $L^\infty$ bound on the source terms in \eqref{holomorphic embeding: structure of the equation} will imply one for the solution. 
This ensures that the $F_{(l)}$ in \eqref{Formal solution of curves} can 
be taken to be small deformations of the flat solutions and, therefore 
and up to restricting $C$ in \eqref{Boundness domain}, that the formal solution 
\eqref{Formal solution of curves} converge.\\

All of this proves the following.

\begin{Theorem}\label{theorem existence of solutions}
    For deformations $\T$ of the twistor space \eqref{Boundness domain} satisfying the assumptions \ref{item holom projection}-\ref{item bounded deformation}, there exists an infinite dimensional space $\HST$ of holomorphic sections $\mathbb{C}^2_0 \hookrightarrow \T$.
    
    The space $\HST$ is an infinite dimensional complex manifold whose local coordinates can be taken to be $x^{I} = (x^{\underline{\dot{\alpha}}\alpha}, x^{\underline{\dot{\alpha}}\alpha(2)}, x^{\underline{\dot{\alpha}}\alpha(3)}, \dots )$. It is the total space of a fibre bundle
    \begin{equation}\label{Projection from HST to M}
   \pi :  \HST \to \M\;,
\end{equation}
where $\M$ is a four-dimensional holomorphic self-dual spacetime, 
the projection being given in local coordinates by 
$x^I \to x^{\underline{\dot{\alpha}}\alpha}\,$.
\end{Theorem}

\subsection{Higher spin frames}

\paragraph{The tangent higher-spin space}

Let $x^I \in \HST$ be a point of the higher spin spacetime, i.e. a solution $F^{\dot{\beta}} = \sum_{k\leq1} F_{(l)}^{\dot{\beta}}$ of \eqref{eq holomorphicity of curves 2} corresponding to a holomorphic plan $X : \mathbb{C}^2_0 \hookrightarrow \T$. Let $x^{\mu}$ be the space time point given by the projection \eqref{Projection from HST to M}, i.e. the holomorphic curve  in the projective twistor space $\hat{X} : \mathbb{C}P^1 \hookrightarrow  \PT$ given by $F^{\dot{\beta}}_{(1)}$. We recall that a tangent vector $V^{\mu}\in T_x\M$ at $x^{\mu}$ is identified with a holomorphic section $\delta F_{(1)}^{\dot{\beta}} \frac{\partial}{\partial \mu^{\dot\beta}}$ of the normal bundle $\hat{N}$ to $\hat{X}$, i.e. satisfying
\begin{equation}
    \bar{\nabla} \delta F_{(1)}^{\dot \beta}=0
\end{equation}
where 
$\bar{\nabla} = \bar{\partial} + \{\h_{(2)},- \}$
is the Dolbeault operator on $\hat{N}\,$.

Similarly, a tangent vector $V^I \in T\HST$ to the point $x^I$ is given by an infinitesimal deformation $\delta F^{\dot{\beta}} = \sum_{k\leq1} \delta F_{(l)}^{\dot{\beta}}$ satisfying the linearized equations. As a result of \eqref{holomorphic embeding: structure of the equation}, they form a triangular inferior system of linear equations
\begin{align}
\begin{array}{lllll}
    \bar{\nabla} \delta F_{1}^{\dot \beta}& & &&\,=0\\
    A_{21}  \delta F_{(1)}^{\dot \beta}  &+ \bar{\nabla} \delta F_{(2)}^{\dot \beta}&  &&=\,0\\
     A_{31}  \delta F_{(1)}^{\dot \beta} &+  A_{32}  \delta F_{(2)}^{\dot \beta} &+ \bar{\nabla} \delta F_{(3)}^{\dot \beta}&&\,=0\\
     A_{41}  \delta F_{(1)}^{\dot \beta} &+ A_{42}  \delta F_{(2)}^{\dot \beta} &+ A_{43}  \delta F_{(3)}^{\dot \beta} &+ \bar{\nabla} \delta F_{(4)}^{\dot \beta} &\,=0\\
     \vdots&&&
     \end{array}
\end{align}
where each of the $\delta F_{(l)}^{\dot{\beta}}$ here defines a section of 
the weighted normal bundle $\hat{N} \otimes \mathcal{O}(l-1)$ to $\hat{X}$ 
and $A_{ij}$ is a homogeneous differential operator of degree $i-j\,$. 
By a suitable linear redefinition of the fields $\delta F_{(i)}^{\dot \beta} \mapsto (\delta F_{(i)}^{\dot \beta})'= \delta F_{(i)}^{\dot \beta} + \sum_{1\leq k < i}  \bar{\nabla}^{-1} A_{ik} \delta F_{(k)} $ one can diagonalize the system. As a result the space of solutions is isomorphic to the space of holomorphic sections of $\underset{1\leq l}{\oplus}\Big( \hat{N} \otimes \mathcal{O}(l-1) \Big)= \underset{1\leq l}{\oplus} \Big(\mathcal{O}(l)\oplus\mathcal{O}(l)\Big)$ and therefore,
\begin{equation}
    \begin{alignedat}{2}
  T_x\HST &= \Gamma\left[X,N\right] \simeq \bigoplus_{1\leq l} \Gamma\left[\hat{X},\hat{N} \otimes \mathcal{O}(l-1)\right]\\
  &= \Gamma\left[\hat{X}, \underset{1\leq l}{\oplus}  \mathcal{O}(l)\right] \oplus \Gamma\left[\hat{X}, \underset{1\leq l}{\oplus}\mathcal{O}(l)\right]
  = \mathcal{O}_0(X) \oplus  \mathcal{O}_0(X) 
\end{alignedat}
\end{equation}
with $\mathcal{O}_0(X)$ the space of holomorphic 
functions on $X = \mathbb{C}^2_0$ vanishing at the origin. In particular, the normal bundle $N\to X$ to the holomorphic plane $X : \mathbb{C}^2_0 \hookrightarrow \T$ 
is trivial, $N \simeq \mathbb{C}^2 \times \mathbb{C}^2_0$, a fact that we will shortly use.

\paragraph{Correspondence space}

The higher-spin spacetime $\HST$ being defined as the set of holomorphic 
sections $X: \IC^2_0 \hookrightarrow \T$, the (infinite dimensional) 
correspondence space is defined as the product manifold  
$\F = \HST \times \IC^2_0\,$. 
We will write the corresponding  coordinates as $(x^I , \lambda^{\alpha})$ 
and in particular the exterior derivative 
${\rm d}= {\rm d}_x +{\rm d}_{\lambda}$ is the sum of a horizontal and 
a vertical exterior derivative:
\begin{align}
    {\rm d}_x&= {\rm d}x^I \frac{\partial}{\partial x^I}\;, & 
    {\rm d}_{\lambda}&= {\rm d}\lambda^{\alpha} 
    \frac{\partial}{\partial \lambda^{\alpha}}\;.
\end{align}
From the definition of $\HST$ as a space of sections, the correspondence space has a natural double fibration, $\pi_1 : \F \to \HST$  and $\pi_3 : \F \to \T$: if \begin{equation}
   X: \left| \lambda_{\alpha} \mapsto \begin{pmatrix} F^{\dot \alpha}(\lambda, \bar{\lambda}) \\ \lambda_{\alpha} \end{pmatrix} \right.
\end{equation} is the map corresponding to the point $x^I$ of the higher spin space then the second fibration is just its evaluation at the point $\lambda$ i.e.
\begin{equation}
   \pi_3: (x^I , \lambda^{\alpha})  \mapsto \begin{pmatrix} F^{\dot \alpha}(\lambda, \bar{\lambda}) \\ \lambda_{\alpha} \end{pmatrix}.
\end{equation}

\paragraph{Higher spin frames}

It follows from the definitions that the derivative 
${\rm d}_xF^{\dot{\alpha}}$ is the map identifying tangent vectors 
$T_x \HST$ at a point $x^I$ of the higher spin space with holomorphic 
sections $\Gamma[X, N]$ of the normal bundle,
\begin{equation}
{\rm d}_x F^{\dot\alpha} \left|
\begin{array}{ccc}
     T\HST&\longrightarrow& \Gamma[X,N]  \\[0.6em]
    V^{I} & \longmapsto &V^I\frac{\partial F^{\dot{\alpha}}}{\partial x^I}(\lambda,\bar{\lambda}) \frac{\partial}{\partial \mu^{\dot{\alpha}}}
\end{array}.
\right. 
\end{equation}
It is useful to write this map in the abstract holomorphic 
local coordinates on the curved twistor space $(z^{\underline{\dot\alpha}})$ 
given by the integrability of the deformed complex structure, 
\begin{align}
	{\rm d}_xF(\lambda,\bar{\lambda})^{\dot\alpha}\partial_{\dot\alpha} 
    = \left( {\rm d}_x F^{\dot\alpha}(H^{-1})_{\dot\alpha}{}^{\underline{\dot\alpha}} \right)(\lambda,\bar{\lambda}) \, \partial_{\underline{\dot\alpha}}
\end{align}
where $(H^{-1})_{\dot\alpha}{}^{\underline{\dot\alpha}}(\lambda,\bar{\lambda})$ 
is the transition map between the basis 
$\frac{\partial}{\partial \mu^{\dot{\alpha}}}$ on the vertical tangent bundle 
$V \to X \simeq \IC^2_0$ induced by the flat complex coordinates 
$(\mu^{\dot{\alpha}}, \lambda_{\alpha})$ and the holomorphic basis 
$\frac{\partial}{\partial z^{\underline{\dot{\alpha}}}}$ induced by 
the deformed complex structure. 
Let us note 
\begin{equation}
   {\rm d}x^I E_{I}{}^{\underline{\dot\alpha}}(x, \lambda_\alpha)
   := {\rm d}_x F^{\dot\alpha}(H^{-1})_{\dot\alpha}{}^{\underline{\dot\alpha}} 
\end{equation}
the corresponding components. In the holomorphic basis, the fact that 
${\rm d}_xF^{\dot\alpha}$ defines an holomorphic section simply reads
\begin{equation}
    \frac{\partial}{\partial \bar{\lambda}^{\dot{\beta}}} 
    \Bigg( {\rm d}_x F^{\dot\alpha}(H^{-1})_{\dot\alpha}{}^{\underline{\dot\alpha}} \Bigg)=0\;.
\end{equation}
Therefore $E^{\underline{\dot\alpha}}(x, \lambda_\alpha) = \sum_k \lambda_{\alpha(k)}E^{\underline{\dot\alpha}\alpha(k)}(x)$ and the differential ${\rm d}_xF$ can be rewritten as
\begin{align}\label{eq spacetime vielbein}
	{\rm d}_xF^{\dot\alpha}\partial_{\dot\alpha} = \left(\sum_{k\in \mathbb{N}} \lambda_{\alpha(k)}E^{\alpha(k) \underline{\dot\alpha}}(x)\right)\, \partial_{\underline{\dot\alpha}}\;.
\end{align}
The higher-spin space $\HST$ is thus equipped with an infinite number of higher-spin frames 
\begin{equation}\label{HS frames}
E^{ \underline{\dot\alpha}\alpha(k)}(x)
= {\rm d}x^I E^{\underline{\dot\alpha}\alpha(k) }_I(x)
\end{equation}
 where $k\in \mathbb{N}$ is related to the helicity by $k=2s-3$.

\subsection{Correspondence with higher-spin self-dual gravity}

\paragraph{The 2-form field of the HS space}
The holomorphic symplectic pairing $\epsilon_{\alpha\beta}$ induced 
on the normal bundle $N \to X\simeq \IC^2_0$ by the Poisson structure 
$\left\{ \cdot , \cdot \right\} = \epsilon^{\dot \alpha \dot \beta} \frac{\partial}{\partial \mu^{\dot \alpha}}\frac{\partial}{\partial \mu^{\dot \beta}}$ 
now allows to construct a two-form on the correspondence space:
\begin{equation}\label{eq favorite object}
   \mathcal{F}(x, \lambda) := {\rm d}_x F^{\dot \alpha} 
   \wedge {\rm d}_x F^{\dot \beta} \epsilon_{\dot{\alpha} \dot{\beta}} 
   = E^{\underline{\dot\alpha}}(x,  \lambda) 
   \wedge E^{\underline{\dot\beta}}(x,  \lambda)
   \epsilon_{\underline{\dot{\alpha}} \underline{\dot{\beta}}} 
   \quad\in \quad \Omega^{2}(\F)\;,
\end{equation}
where we made use of the fact that the deformation of the complex structure 
must preserve the Poisson structure $\epsilon_{\dot{\alpha} \dot{\beta}} = \epsilon_{\underline{\dot{\alpha}} \underline{\dot{\beta}}}\,$. 
Making use of the higher spin frames \eqref{HS frames} 
this 2-form can be rewritten as
\begin{align}\label{2-form, twistor}
	\mathcal{F}(x, \lambda_\alpha) 
	= \left( \sum_s \lambda_{\alpha(s)}E^{\alpha(s)\underline{\dot\alpha}}(x)\right) 
    \wedge \left( \sum_{s'} \lambda_{\beta(s')}E^{\beta(s') 
    \underline{\dot\beta}}(x) \right) 
    \epsilon_{\underline{\dot\alpha}\underline{\dot\beta}}\;.
\end{align}
By construction the 2-form \eqref{2-form, twistor} is horizontally closed 
and simple:
\begin{align}
   {\rm d}_x\mathcal{F} &=0\;, &\mathcal{F} &= 2 \,E^{\underline{\dot{0}}}(x,  \lambda) \wedge E^{\underline{\dot{1}}}(x,  \lambda)\;.
\end{align}
The second equation trivially implies
\begin{equation}
    \mathcal{F}  \wedge \mathcal{F}  =0\;,
\end{equation}
while the closure condition means that there must exist a horizontal 1-form 
$\mathcal{A}= {\rm d}x^I \mathcal{A}_I(x,\lambda)$ on the correspondence space 
such that
\begin{equation}
    \mathcal{F} = {\rm d}_x \mathcal{A}\;.
\end{equation}
Since $\mathcal{F}$ is polynomial in $\lambda$ we must have
\begin{equation}\label{Nonlinear graviton thrm: potential in the correspondence space}
    \mathcal{A}= \sum_s {\rm d}x^I \mathcal{A}_I{}^{\alpha(s)}(x) 
    \lambda_{\alpha(s)}\;,
\end{equation}
and we thus see that the fields $\mathcal{A}{}^{\alpha(s)}(x)$ play the 
role of potentials for the higher spin frames 
$E^{\dot{\alpha}\alpha(s)}(x)\,$.\\

\paragraph{Choice of gauge}

Recall that we have a canonical projection 
\begin{align}
 \pi :  \HST \to \M\;.
\end{align}
In the other direction, there is no canonical way to embed  the space 
$\M$ in $\HST$ : a solution to the homogeneity $l=1$ part of the equations 
\eqref{eq holomorphicity of curves} does not extend  to a unique solution 
in the superior homogeneities.
To recover a solution of higher-spin self-dual gravity \eqref{Field equations of SD HS gravity}, one needs to
choose a section
\begin{align}\label{eq choice of section}
    s : \M \hookrightarrow \HST\;.
\end{align}

This section can be extended in a trivial way to the correspondence spaces
\begin{align}
    \tilde{s}: \M \times \IC^2_0 \rightarrow \HST \times \IC^2_0
\end{align}
by defining $\tilde{s}(x^\mu, \lambda_\alpha) = (s(x^\mu), \lambda_\alpha)\,$. 
It is then clear that by pulling back the potential 
\eqref{Nonlinear graviton thrm: potential in the correspondence space} 
by such a choice of embedding the resulting field $A = \tilde{s}^*\mathcal{A}$ 
will be a solution to higher-spin self-dual gravity: the potentials
\begin{equation}\label{HS field: pull back of potential}
    {\rm d}x^{\mu} A_{\mu}^{\alpha(s)} = s^*\left(\mathcal{A}{}^{\alpha(s)} \right) 
\end{equation}
will automatically satisfy \eqref{Field equations of SD HS gravity} as a result of
\begin{equation}
    \tilde{s}^*\left( \mathcal{F} \wedge \mathcal{F} \right) =0\;.
\end{equation}

We therefore proved an extension of the nonlinear graviton theorem for higher-spin self-dual gravity:
\begin{Theorem}\label{theorem NLG}
There is a 1-1 correspondence between, small solutions of self-dual higher spin gravity with zero cosmological constant and small deformations of the complex structure of $\mathbb{T}$ satisfying the assumptions \ref{item holom projection}-\ref{item bounded deformation} together with a choice of section \eqref{eq choice of section}.
\end{Theorem}
This theorem is an immediate consequence of Theorem \ref{theorem existence of solutions} and the previous discussion. Since our twistor space is restricted to the open set \eqref{Boundness domain}, it holds for solution of higher-spin self-dual gravity in a neighbourhood of the origin. In the next section we will prove that the choice of section appearing in this theorem is just a gauge choice:
\begin{Theorem}\label{theorem NLG with gauge}
There is a 1-1 correspondence between small solutions of higher spin self-dual gravity with zero cosmological constant \underline{up to gauge} and small deformations of the complex structure of $\mathbb{T}$ satisfying the assumptions \ref{item holom projection}-\ref{item bounded deformation}.
\end{Theorem}

\section{Higher-spin space geometry and higher-spin symmetries}\label{section: Higher-spin space geometry and higher-spin symmetries}

\subsection{Higher-spin space geometry}

Before discussing how higher-spin symmetries are realized in the present context, we briefly recapitulate the higher-spin geometry which emerged from our twistor space realization. The flat model is detailed below as a concrete illustration.

\paragraph{Higher-spin space geometry} From Theorem \ref{theorem existence of solutions} the higher-spin space $\HST$ is an infinite dimensional complex manifold with coordinates $x^{I} = (x^{\underline{\dot{\alpha}}\alpha}, x^{\underline{\dot{\alpha}}\alpha(2)}, x^{\underline{\dot{\alpha}}\alpha(3)}, \dots )$. It naturally is the total space of a fibre bundle $\HST \to \M$ over a four dimensional holomorphic self-dual spacetime, the projection being given by $x^I \to x^{\mu}= x^{\underline{\dot{\alpha}}\alpha}$. It comes equipped with higher-spin frames, $k\in \mathbb{N}$ being related to the helicity via $k=2s-3$,
\begin{equation}\label{FrameSDHS}
 E^{\underline{\dot\alpha}\alpha(k)} 
 = {\rm d}x^I E^{\underline{\dot\alpha}\alpha(k)}_I(x)\;,
\end{equation}
satisfying the differential equation 
\begin{align}
    {\rm d}_x \mathcal{F}^{\alpha(k)} &=0\;,& \textrm{where}\quad 
    \mathcal{F}^{\alpha(k)} = 
    \sum_{l+m=k}E^{\underline{\dot\alpha}\alpha(l)} \wedge 
    E^{\underline{\dot\beta}\alpha(m)} \epsilon_{\dot{\underline{\alpha}} 
    \dot{ \underline{\beta}}}\;.
\end{align}
The higher-spin frame derive from potential 1-form 
$\mathcal{A}^{\alpha(k)}$ via the relation 
$\mathcal{F}^{\alpha(k)} = {\rm d}_x \mathcal{A}^{\alpha(k)}\,$. 
The fields and equations of higher-spin self-dual gravity are obtained by making a choice of embedding 
of spacetime into the higher-spin space $s: \M \to \HST$ (more precisely choosing 
a section of the bundle $\HST \to \M$) and pulling back. 
As we shall see in the coming section, this freedom in realizing spacetime 
inside the higher-spin space is the way higher-spin symmetries are realized 
in the present context.

\paragraph{The flat model}

In the flat model the holomorphic planes are given by \eqref{Flat model: holomorphic planes} 
with $\{ X^{\dot \alpha \alpha(k)}\} \to X^{\dot \alpha \alpha}$ playing the 
role of coordinates on $\HST \to \M\,$. 
The higher-spin frames are simply 
$E^{\dot \alpha \alpha(k)} = {\rm d}X^{\dot \alpha \alpha(k)}\,$, 
and are generated by the potentials
\begin{equation}
    \mathcal{A}^{\alpha(k)} =  \sum_{l+m=k} X^{\dot \alpha \alpha(l)}
    \,{\rm d}X^{\dot \beta \alpha(m)} \epsilon_{\dot \alpha \dot \beta}\;.
\end{equation}
A choice of section $s:\M\to \HST$ reads
\begin{equation}
    s: x^{\alpha\dot\alpha} \mapsto F^{\dot \alpha}(x) = x^{\alpha\dot\alpha} \lambda_{\alpha}  + \sum_{k\geq 2} \varphi^{\dot \alpha \alpha(k)}(x) \lambda_{\alpha(k)}
\end{equation}
and we therefore obtain the higher-spin self-dual gravity fields via the pull-back
\begin{equation}
    \begin{alignedat}{3}
    A^{\alpha(2)}(x) &= x^{\dot \alpha \alpha}\,{\rm d}x^{\dot \beta \alpha} \epsilon_{\dot \alpha \dot \beta}\;,\\
     & ~\;\vdots\\
    A^{\alpha(k)}(x) &=  {\rm d}x^{\gamma \dot\gamma} 
    \left(\;\sum_{l+m=k} \varphi^{\dot \alpha \alpha(l)} 
    \frac{\partial \varphi^{\dot \beta \alpha(m)}}{\partial x^{\gamma \dot \gamma}} 
    \,\epsilon_{\dot \alpha \dot \beta} \right)\;.
\end{alignedat}
\end{equation}
The spin-2 field here coincide with the potential 
$A^{\alpha(2)}_0(x) := x^{\dot \alpha \alpha}{\rm d}x^{\dot \beta \alpha} \epsilon_{\dot \alpha \dot \beta}$ 
for the Minkowski frame field 
$e^{\alpha\dot \alpha}_0 = {\rm d}x^{\alpha \dot \alpha}$ 
and the higher-spin fields are pure gauge in the sense that they can be put to 
zero\footnote{This can also be seen from the fact that, in correspondence space, 
$A(x,\lambda) = \sum_{2\leq k} A^{\alpha(k)}(x) \lambda_{\alpha(k)}$ is obtained 
from $A_0 = A^{\alpha(2)}_0(x) \lambda_{\alpha(2)}$ via 
\eqref{eq action of gauge symmetry on A} and taking 
$\varphi^{\dot \alpha\alpha}(x,\lambda) = \sum_{1\leq k}\varphi^{\dot \alpha \alpha(k)}(x)\lambda_{\alpha(k-1)}\,$.} 
by choosing $\varphi^{\dot \alpha \alpha(l)}=0$ for all $l \geq2\,$. 
Deforming the section $s$ infinitesimally by 
$\delta \varphi^{\dot \alpha \alpha(k)}$ is then equivalent to a higher-spin 
symmetry \eqref{HS SD gravity: gauge} since
\begin{equation}
\begin{alignedat}{3}
        \delta A^{\alpha(k)}(x) &=  {\rm d}x^{\gamma \dot\gamma} \sum_{l+m=k}  \left(\; \delta\varphi^{\dot \alpha \alpha(l)} \frac{\partial \varphi^{\dot \beta \alpha(m)}}{\partial x^{\gamma \dot \gamma}} + \varphi^{\dot \alpha \alpha(l)} \frac{\partial \delta\varphi^{\dot \beta \alpha(m)}}{\partial x^{\gamma \dot \gamma}} \right) \epsilon_{\dot \alpha \dot \beta},\\
        &= {\rm d} \left( \sum_{l+m=k}  \varphi^{\dot \alpha \alpha(m)}\delta\varphi^{\dot \beta \alpha(l)} \epsilon_{\dot \alpha \dot \beta}\right) +2{\rm d}x^{\gamma \dot\gamma} \sum_{l+m=k}  \; \delta\varphi^{\dot \alpha \alpha(l)} \frac{\partial \varphi^{\dot \beta \alpha(m)}}{\partial x^{\gamma \dot \gamma}}
        \epsilon_{\dot \alpha \dot \beta}\;,
\end{alignedat}
\end{equation}
where, to match with \eqref{HS SD gravity: gauge}, 
one only has to recognize that the second term is of the form
\begin{equation}
\begin{alignedat}{2}
        2{\rm d}x^{\gamma \dot\gamma} \sum_{n+p=k}  \; \eta^{\dot\delta  \delta\alpha(n)}  F_{\dot\delta \delta \dot \gamma \gamma}^{\alpha(p)} &= 2{\rm d}x^{\gamma \dot\gamma} \sum_{n+p=k}  \; \eta^{\dot\delta \delta \alpha(n)}   \sum_{q+m=p} \frac{\partial \varphi^{\dot \alpha \alpha(q)}}{\partial x^{\delta \dot \delta}}\frac{\partial \varphi^{\dot \beta \alpha(m)}}{\partial x^{\gamma \dot \gamma}} \epsilon_{\dot \alpha \dot \beta} \\
    &= 2{\rm d}x^{\gamma \dot\gamma} \sum_{n+q+m=k}  \; 
    \eta^{\dot\delta \delta\alpha(n)}    \frac{\partial \varphi^{\dot \alpha \alpha(q)}}{\partial x^{\delta \dot \delta}}\;
    \frac{\partial \varphi^{\dot \beta \alpha(m)}}{\partial x^{\gamma \dot \gamma}} \epsilon_{\dot \alpha \dot \beta}
\end{alignedat}
\end{equation}
which is obtained by the identification\footnote{The map $\{\eta^{\dot\alpha \alpha(l)}\}  \mapsto \{\delta\varphi^{\dot \alpha \alpha(l)}\}$ is here an isomorphism since it can be rewritten as \begin{equation*}
\delta\varphi^{\dot \alpha \alpha(l)} =  \eta^{\dot\alpha \alpha(l)} +  \sum_{n < l-1,\; n+q=l}  \; \eta^{\dot\delta \delta \alpha(n)}    \frac{\partial \varphi^{\dot \alpha \alpha(q)}}{\partial x^{\dot \delta \delta}}    
\end{equation*} 
which can be inverted order by order in $l$.}
\begin{equation}
\begin{alignedat}{2}
\delta\varphi^{\dot \alpha \alpha(l)} = \sum_{n+q=l}  \; \eta^{\dot\delta \delta \alpha(n)}    \frac{\partial \varphi^{\dot \alpha \alpha(q)}}{\partial x^{\dot \delta \delta}}\;.
\end{alignedat}
\end{equation}

\subsection{Gauge symmetries}\label{ssection: gauge symmetries}

The goal of this section is to prove Theorem \ref{theorem NLG with gauge}.
Theorem \ref{theorem NLG} already gives us the correspondence
between complex deformations $\T$ of twistor space, together with choice of embedding \eqref{eq choice of section},
and solutions to higher-spin self-dual gravity \eqref{Field equations of SD HS gravity}. The remaining content of Theorem \ref{theorem NLG with gauge} states the equivalence of two (a priori) different notions of symmetries on the solution space we introduced in the previous sections:

\begin{enumerate}[(a)]
\item Horizontal, holomorphic diffeomorphisms $\Phi : \M \times \IC^2_0 \rightarrow \M\times \IC^2_0$ acting on higher-spin self-dual gravity solutions via \eqref{eq action of gauge symmetry on A} and resulting in the higher spin symmetry \eqref{HS SD gravity: gauge};
\label{item gauge symmetry diffeo}
\item  Choices of embedding $s : \M \hookrightarrow \HST$ of the type \eqref{eq choice of section} giving rise to higher-spin self-dual gravity solutions.
\label{item gauge symmetry section}
\end{enumerate}

At this stage, the fact than any symmetry of the type \ref{item gauge symmetry diffeo} can be seen as a symmetry of the type 
\ref{item gauge symmetry section} is easy to see and prove:
The double fibration on the correspondence space, see Figure \ref{Fig: Double fibration}, defines a 4-parameter family of holomorphic sections \eqref{Sections of T}. As we already pointed out, diffeomorphisms of the form \eqref{Diffeo in correspondence space} on the correspondence space then act on this 4-parameter family, via \eqref{eq action of the gauge symmetry on the curves}, producing a new 4-parameter family of holomorphic sections \eqref{Sections of T}:
\begin{align}
    F^{\dot{\beta}}(x^\mu, \lambda_\alpha) \qquad\mapsto\qquad  F^{\dot{\beta}}(\varphi^{\mu}(x, \lambda), \lambda_\alpha)
\end{align}
These two families are solutions of the same equation \eqref{eq holomorphicity of curves} (thanks to the assumption that $\varphi$ is holomorphic in $\lambda$) and define distinct sections of \eqref{Projection from HST to M}. One can then check that, consistently, the HS fields obtained on $\M$ via \eqref{HS field: pull back of potential} in this way are in fact gauge equivalent. Gauge equivalent solutions \eqref{eq action of gauge symmetry on A} are thus related to choices of different embeddings $\M \hookrightarrow \HST$.

The converse, i.e. the fact that any symmetry of the type 
\ref{item gauge symmetry section} can be seen as a symmetry
of the type \ref{item gauge symmetry diffeo} is much less obvious. 
The reason is that the space of possible embeddings
$s : \M \hookrightarrow \HST$ has little structure 
and is thus difficult to describe non linearly. The following Lemma 
will give a proof at infinitesimal level.

\begin{Lemma}\label{Lemma: gauge invariance}
	Two different choices of embeddings $\M \hookrightarrow \HST$ correspond (in the sense of Theorem \ref{theorem NLG}) to HS gauge fields which are gauge equivalent.
\end{Lemma}

\begin{proof}
    
Two infinitesimally close, sections $i : \M\hookrightarrow\HST$ of \eqref{Projection from HST to M} can be parametrized as
\begin{align}\label{eq: two different choices of embeddings}
	      x^{\mu}  &\mapsto  \bm{x}^I(x^{\mu})& x^{\mu} \mapsto  \bm{x}^I(x^{\mu}) + \delta \bm{x}^I(x^{\mu})\;.
\end{align}
where
\begin{align}
	      \delta \bm{x}^I(x^{\mu})=: \bm{\eta}^I(x^{\mu})
\end{align}
can always be though of as the restriction of an extended vector field 
$\bm{\eta}^I(\bm{x})$ on whole $\HST$.
The proof will be independent of the latter extension. 
Now from the perspective of $\HST$, this vector field can be identified \emph{uniquely}, through ${\rm d}_{\bm{x}}F^{\dot\alpha}$, 
with a holomorphic section of the normal bundle over the corresponding curve 
$X : \mathbb{C}_0^2 \hookrightarrow \T$ in twistor space :
\begin{align}\label{eq: from eta I to section}
	\bm{\eta}^I(\bm{x})\partial_{\bm{x}^I} \mapsto \iota_{\bm{\eta}}{\rm d}_{\bm{x}} F^{\dot\alpha} \partial_{{\dot\alpha}} =: \eta^{\dot\alpha}(\bm{x}, \lambda)\partial_{{\dot\alpha}} =\sum_{s\geq 1}\eta^{\underline{\dot\alpha}\alpha(s)}(\bm{x}^I)\lambda _{\alpha(s)}
	\partial_{\underline{\dot\alpha}}\;,
\end{align}
where in the last equality we used an holomorphic basis, see \eqref{eq spacetime vielbein}.

Now,  under the change of embedding \eqref{eq: two different choices of embeddings}, 
the two-form 
$\mathcal{F} = {\rm d}_{\bm{x}}F^{\dot\alpha}\wedge {\rm d}_{\bm{x}}F_{\dot\alpha}$ 
is transformed as
\begin{align}\label{eq symmetry infini dim}
	\mathcal{L}_{\bm{\eta}}(\mathcal{F})  := \mathcal{L}_{\bm{\eta}}({\rm d}_{\bm{x}}F^{\dot\alpha}\wedge {\rm d}_{\bm{x}}F_{\dot\alpha}) 
    =  2\,{\rm d}_{\bm{x}}\left( \iota_{\bm{\eta}} 
    {\rm d}_{\bm{x}}F^{\dot\alpha}\right)\wedge {\rm d}_{\bm{x}}F_{\dot\alpha}
	= 2\,{\rm d}_{\bm{x}}{\eta}^{\dot\alpha}\wedge {\rm d}_{\bm{x}}F_{\dot\alpha}\;,
\end{align}
where $\eta^{\dot\alpha}$ is defined by \eqref{eq: from eta I to section}.

On the other hand, once the choice of embedding $i$ has been made, one can act on the two-form $i^*\mathcal{F}$ with a diffeomorphisms \eqref{eq action of gauge symmetry on A}. The corresponding infinitesimal action is:
\begin{align}\label{eq symmetry four dim}
	\delta_{\xi} (i^*\mathcal{F}) = {\rm d}_x\left(\iota_\xi i^*\mathcal{F}\right)
	= 2\,{\rm d}_x((\iota_{\xi}(i^*{\rm d}_{\bm{x}}F^{\dot\alpha})) 
    \;i^*{\rm d}_{\bm{x}}F_{\dot\alpha})
    = 2\,{\rm d}_x(\iota_{\xi}(i^*{\rm d}_{\bm{x}}F^{\dot\alpha}))
    \wedge i^*{\rm d}_{\bm{x}}F_{\dot\alpha}\;,
\end{align}
with $\xi(x^\mu,\lambda) = \sum_{l\geq 0} \xi^{\mu \alpha(l)}(x) \lambda_{\alpha(l)} \partial_{\mu}\,$.

We would like to relate the two type of transformations \eqref{eq symmetry infini dim} and 
\eqref{eq symmetry four dim} in such a way that the following diagram commute :
\begin{equation}\label{eq commutative diagram}
	\begin{tikzcd}
      \mathcal{F} \ \arrow[r, "\bm{\eta}", shorten=-1mm, end anchor={[xshift=6ex]west}] \arrow[d, "i^*"] &  \quad \quad \quad \mathcal{L}_{\bm{\eta}} \mathcal{F}  \arrow[d, "i^*", shift left=4.5ex]  \\
	  i^*\mathcal{F}  \arrow[r, "\xi"]  & \delta_\xi(\iota^* \mathcal{F}) = \iota^*(\mathcal{L}_\eta\mathcal{F})  
\end{tikzcd}
\end{equation}
Comparing \eqref{eq symmetry four dim} and \eqref{eq symmetry infini dim} one sees that one needs to have
\begin{align}\label{eq isomorphism of symmetries}
	i^*{\eta}^{\dot\alpha} = \iota_{\xi}(i^*{\rm d}_{\bm{x}}F^{\dot\alpha})\;.
\end{align}
In order to prove Lemma \ref{Lemma: gauge invariance} 
(and thus Theorem \ref{theorem NLG with gauge}), we therefore have to show 
that for every infinitesimally close choice of embedding
\eqref{eq: two different choices of embeddings} (i.e. for every
$\bm{\eta}$), there exists a symmetry
$\xi$ satisfying \eqref{eq isomorphism of symmetries}.

To see that it can always be done, make use of \eqref{eq spacetime vielbein} and \eqref{eq: from eta I to section} to rewrite \eqref{eq isomorphism of symmetries} as
\begin{align}\label{eq isomorphism 2}
	\eta^{\underline{\dot\alpha}\alpha(s)}(x^\mu)\lambda_{\alpha(s)}
	= \sum_{k+l = s}\xi^{\mu\alpha(l)}E_\mu^{\underline{\dot\alpha}\alpha(k)}(x^\mu)
	\lambda_{\alpha(k+l)},\qquad\qquad \forall s\in \mathbb{N}.
\end{align}
This can then be solved recursively, order by order in $s$:
For $s= 1$, equation \eqref{eq isomorphism 2} is
\begin{align}
	\eta^{\dot\alpha\alpha} = \xi^\mu E_\mu^{\dot\alpha\alpha}, 
\end{align}
which can be inverted because, from the usual nonlinear graviton theorem, $E_\mu^{\dot\alpha\alpha}$ is the vielbein for the spin $2$ field.

For the equation at order $s$, assuming all the lower
spins have been solved, i.e. that all the 
$\xi^{\mu\alpha(l)}$ are known 
for $l \leq s-2$, then equation \eqref{eq isomorphism 2} is
\begin{align}
	\eta^{\dot\alpha\alpha(s)} = 
	\xi^{\mu\alpha(s-1)}E_\mu^{\dot\alpha\alpha}
	+ \sum_{l \leq s-2, \ k+l = s} \xi^{\mu\alpha(l)} E_\mu^{\dot\alpha\alpha(k)},
\end{align}
from which we can express $\xi^{\mu\alpha(s-1)}$
thanks to the invertibility of $E_\mu^{\dot\alpha\alpha}
$.
\end{proof}

\subsection{Other notions of higher-spin spacetime in the literature.}

There are not many consistent, interacting higher-spin gauge theories
for propagating fields in a spacetime of Lorentzian signature.
A large body of works, including the construction of exact solutions
to interacting higher-spin field equations, has rested upon Vasiliev's 
model as formulated in \cite{Vasiliev:1990en}. 
The latter model is fully consistent\footnote{It however shows a 
severe nonlocality issue when one tries to express the nonlinear dynamics 
perturbatively around (A)dS$_4$, as shown in \cite{Boulanger:2015ova}.
That the theory should be nonlocal in spacetime had already been foreseen in 
\cite{Vasiliev:1999ba}, but contrary to String Theory where there are consistent
and local low-energy supergravity limits, Vasiliev's theory does not seem 
to admit any consistent interacting limit which would be local in spacetime, 
and at the same time which would preserve its infinite-dimensional higher-spin 
symmetry. For related discussions about possible and tractable 
higher-spin symmetry breaking mechanisms, see \cite{Bekaert:2010hw}.}
and admits in its free limit an infinite spectrum of massless higher-spin
gauge fields of unbounded spin propagating on (anti) de Sitter 
four-dimensional spacetime. 

In Appendix \ref{Appendix:Vasiliev} we propose a critical review 
of the infinite-dimensional higher-spin spacetime proposed by the authors 
of \cite{Sezgin:2011hq}, and make a comparison with the construction 
obtained in the present paper for self-dual higher-spin gravity.

In the Cartan-like reformulation of Vasiliev's equations 
proposed in \cite{Sezgin:2011hq}, the way the usual spacetime 
$\cal M$ is embedded in the infinite-dimensional higher-spin spacetime
${\cal M}_{HS}$ is determined by the solution, at $Z=0$ -- where 
$Z$ collectively denotes the non-commutative twistor-like coordinates 
in Vasiliev's formulation \cite{Vasiliev:1990en} -- of the field 
equations for the various fields, and this should be compared with 
the space of possible embeddings $s : \M \hookrightarrow \HST$ in 
self-dual higher-spin gravity. 
The latter space has little structure at the nonlinear level, 
as explained before Lemma \ref{Lemma: gauge invariance}. 
In the same way, for Vasiliev's theory the embedding 
$\M \hookrightarrow {\cal M}_{HS}$ is also complicated to describe 
because of the high degree of nonlinearity of Vasiliev's equation, 
not to mention its nonlocal nature.

\section{Plebanski potentials}\label{section: Plebanski potentials}

In this section we explain how the Plebanski potential formulation of higher-spin self-dual gravity from \cite{Monteiro:2022xwq} can be recovered from our twistor space construction.

The main motivation for this, apart from being a reality check, is that, as we shall see, the derivation extends straightforwardly to the Plebanski formulation of $chs(0)$ (the Poisson contraction of Chiral higher spin gravity, see \cite{Monteiro:2022xwq}): negative helicity fields  then lead to  singularities of the holomorphic planes $\mathbb{C}_0^2 \hookrightarrow \T$ as $\lambda_\alpha\rightarrow 0$. This strongly suggests the existence of a nonlinear graviton theorem relating generic twistor space deformations, i.e. not assuming \ref{item bounded deformation}, to solutions of $chs(0)$. See more on this in the final discussion.

\subsection{Plebanski potentials for spin $s\geq 2$ fields}\label{subsec hdSDG potentials}
\label{subsec: Plebanski formulation for Higher spins SDG}

We will here recover the Plebanski potential formulation of higher-spin self-dual gravity \cite{Monteiro:2022xwq} from our twistor space construction: We start by showing how these potentials can be derived from a natural gauge fixing of the holomorphic planes in twistor space and we then recover the fields equations for every positive helicity $s$ field.

An integrable $\bar \partial$-operator can always be trivialized on a Stein (suitably convex)  open set by introducing holomorphic coordinates on that set.  These coordinates can be smoothly extended away from that set using bump functions to provide new coordinates for which the $\bar{d}$-operator is trivial on the original open set.  We now do so on some small neighborhood $V_{\lambda_2=0 }$ of $\lambda_2=0$.  On this set the  holomorphic plane $F^{\dot\alpha}$ will be holomorphic and so we expand in a Taylor series in $\lambda_2$ around $\lambda_2=0$. 

In general $F^{\dot\alpha}$ will be expanded as
\begin{align}\label{eq spin s expansion of F}
	F^{\dot\alpha} = \lambda_1z^{\dot\alpha} + 
	\lambda_2 w^{\dot\alpha} +
	\sum_{s\geq 2} \left(\frac{\lambda_2^s}{\lambda_1}\theta_s^{\dot\alpha}
	+ \frac{\lambda_2^{s+1}}{\lambda_1^2}\varphi_s^{\dot\alpha}
	+ \frac{\lambda_2^{s+2}}{\lambda_1^3}\chi_s^{\dot\alpha} \right)
	+...
\end{align}

 We know from the previous sections that in the SD-HSGR theory we can rewrite this holomorphic two-form as a 
sum on the different homogeneities in $\lambda_\alpha$ :
\begin{align}
	{\rm d}_xF^{\dot\alpha}\wedge {\rm d}_xF_{\dot\alpha}= 
\sum_{s}\lambda_{\alpha(s)}
\Sigma^{\alpha(s)}_{[2]}
\end{align}
where 
\begin{align}
	\Sigma^{\alpha(s)} = \sum_{s_1,s_2 ~s.t.~ s_1+s_2=s}
	E^{\alpha(s_1)\underline{\dot\alpha}}(X)\wedge 
	E^{\beta(s_2)\underline{\dot\beta}}(X) \epsilon_{\underline{\dot\alpha}\underline{\dot\beta}}.
\end{align}
Therefore we can analyse ${\rm d}_xF^{\dot\alpha}\wedge {\rm d}_xF_{\dot\alpha}$
for each homogeneity $s$ separately and then sum over all of them. 
We can compute that the homogeneity $s>2$ terms are given by
\begin{align}
	\lambda_{\alpha(s)}\Sigma^{\alpha(s)}
	& =  \ \lambda_2^s(2 {\rm d}z^{\dot\alpha} \wedge {\rm d}\theta_{s\dot\alpha})
	+ \frac{\lambda_2^{s+1}}{\lambda_1} \left(2 {\rm d}z^{\dot\alpha}\wedge 
    {\rm d}\varphi_{s\dot\alpha} + 2 {\rm d}w^{\dot\alpha} 
    \wedge {\rm d}\theta_{s\dot\alpha}\right) \\
	 + & \ \frac{\lambda_2^{s+2}}{\lambda_1^2} 
	\left(2 {\rm d}z^{\dot\alpha} \wedge {\rm d}\chi_{s\dot\alpha}
	+ 2 {\rm d}w^{\dot\alpha}\wedge {\rm d}\varphi_{s\dot\alpha}
	+ \sum_{\overset{s_1,s_2 \text{ st }}{s_1+s_2= s+2}}
	2 {\rm d}\theta_{s_1}^{\dot\alpha} \wedge {\rm d}
    \theta_{s_2\dot\alpha} \right) \\
	+ & \  o (\lambda_1^{-2})\;.
\end{align}

This two-form being holomorphic and global, 
we impose that the poles in $\lambda_1$ vanish, and this gives the two conditions 
\begin{align}
	{\rm d}z^{\dot\alpha}\wedge {\rm d}\varphi_{s\dot\alpha} 
    + {\rm d}w^{\dot\alpha} \wedge {\rm d}\theta_{s\dot\alpha} = 0
	\label{eq spin s first vanishing cond} \\
	2 {\rm d}z^{\dot\alpha}\wedge {\rm d}\chi_{s\dot\alpha}
	 + 2 {\rm d}w^{\dot\alpha} \wedge {\rm d}\varphi_{s\dot\alpha}
	 + \sum_{\overset{(s_1,s_2) \text{ st }}{s_1+s_2- s=2}}
	{\rm d}\theta_{s_1}^{\dot\alpha} \wedge {\rm d}\theta_{s_2\dot\alpha} =0\;.
	 \label{eq spin s second vanishing cond}
\end{align}
The first condition \eqref{eq spin s first vanishing cond} is solved by the existence of potentials $\theta_s$, for each $s$,
that we will call the \emph{Plebanski potentials} for the spin $s$ field :
\begin{align}
	\eqref{eq spin s first vanishing cond} \iff 
	\exists \;\theta_s, \varphi_s \text{ st. } 
	\theta_{s \dot\alpha} = \frac{\partial \theta_s}{\partial w^{\dot\alpha}} \  \text{ and } \ \varphi_{s\dot\alpha} = \frac{\partial \varphi_s}{\partial z^{\dot\alpha}} \notag
	\\ \text{ and }\ \left( \frac{\partial^2 \theta_s}{\partial z^{\dot\beta}\partial w_{\dot\alpha}} = \frac{\partial^2 \varphi_s}{\partial w_{\dot\alpha}\partial z^{\dot\beta}}
	\iff  \frac{\partial\theta_{s\dot\alpha}}{\partial z^{\dot\beta}} - \frac{\partial\varphi_{s\dot\beta}}{\partial w^{\dot\alpha}} = 0 \right) \label{eq plebanski potentials}
.
\end{align}
Taking the ${\rm d}w\wedge {\rm d}w$ part of the second condition 
\eqref{eq spin s second vanishing cond}, and
using \eqref{eq plebanski potentials}, one recovers
\begin{align}
	 2 \frac{\partial^2\theta_{s}}{\partial z^{[\dot\delta}\partial w^{\dot\gamma]}} +
	\sum_{\overset{(s_1,s_2) \text{ st }}{s_1+s_2-s=2}}
	 \frac{\partial^2\theta_{s_1}}{\partial w^{[\dot\delta}\partial w_{\dot\alpha}}\frac{\partial^2\theta_{s_2}}{\partial w^{\dot\gamma]}\partial w^{\dot\alpha}} = 0\;.
\end{align}

This is exactly the same formula as in \cite{Monteiro:2022xwq} for higher-spin self-dual gravity : 
\begin{align}
	\square \phi_s +\sum_{\overset{s_1,s_2 \geq 2}{s_1 +s_2 -s = 2}} \phi_{s_1} \overset{\leftrightarrow}{P^2} \phi_{s_2} = 0\;.
\end{align}
with
\begin{align}
	\square & := \epsilon^{\alpha\beta}\epsilon^{\dot\alpha\dot\beta}\frac{\partial}{\partial x^{\dot\alpha\alpha}}\frac{\partial}{\partial x^{\dot\beta\beta}} = 2 \epsilon^{\dot\alpha\dot\beta}
	\frac{\partial}{\partial z^{\dot\alpha}}\frac{\partial}{\partial w^{\dot\beta}} \\
	\text{ and }\quad f\ \overset{\leftrightarrow}{P}^2 g
	& := \epsilon^{\dot\alpha\dot\gamma}\epsilon^{\dot\beta\dot\delta}\partial_{w^{\dot\alpha}}\partial_{w^{\dot\beta}}f
	\partial_{w^{\dot\gamma}}\partial_{w^{\dot\delta}}g\;.
\end{align}

\subsection{Plebanski potentials for spin $s < 2$ fields}

From \cite{Monteiro:2022xwq}, one sees that the equations of motion for the spin $s$ field of the HS theory
chs(0) have the
exact same form, namely:
\begin{align}
	\square \phi_s +\sum_{\overset{s_1,s_2 st.}{s_1 +s_2 -s = 2}} \phi_{s_1} \overset{\leftrightarrow}{P^2} \phi_{s_2} = 0\;,
\end{align}
but now the spins are allowed to be smaller than 2.
Furthermore, in the computations above concerning the higher-spin self-dual gravity Plebanski potentials, we have actually never used the fact 
that $s, s_1,s_2$ needed to be positive or greater than 2.

This means that we can extend our ansatz for the holomorphic plane \eqref{eq spin s expansion of F} to all spins, positive and negative :
\begin{align}\label{eq spin s negative expansion of F}
	F^{\dot\alpha} = \lambda_1z^{\dot\alpha} + 
	\lambda_2 w^{\dot\alpha} +
	\sum_{s\in \mathbb{Z}} \left(\frac{\lambda_2^s}{\lambda_1}\theta_s^{\dot\alpha}
	+ \frac{\lambda_2^{s+1}}{\lambda_1^2}\varphi_s^{\dot\alpha}
	+ \frac{\lambda_2^{s+2}}{\lambda_1^3}\chi_s^{\dot\alpha} \right)
	+o(\lambda_1^{-3})\;.
\end{align}

Proceeding in the exact same way as section \ref{subsec hdSDG potentials} 
will give the same formulae for homogeneity $s$ terms, even with $s<2$. 
With this we recover the same equations, extended to all $s \in 
\mathbb{Z}$, for the theory called chs(0) in \cite{Monteiro:2022xwq}.
This theory including both positive and negative spins is also called Poisson 
chiral higher spin theory \cite{Ponomarev:2017nrr}.

The drastic change lies in the geometrical interpretation of the expansion \eqref{eq spin s negative expansion of F}. 
In fact, when considering negative homogeneities, we will still
ask for the 
vanishing of the poles in $\lambda_1$, producing the same expected equations, but this time we can see from \eqref{eq spin s negative expansion of F} that the fields with $s<0$ will inevitably introduce poles in $\lambda_2$. 
It seems therefore that one needs to choose a specific direction, 
$\lambda_1$ or $\lambda_2$, in which the two-form will have singularities.
The physical interpretation of these singularity still needs to be investigated, as well as the twistor theory behind it.

\section{Discussion}

In this article we studied a generalisation of Penrose's nonlinear graviton theorem \cite{Penrose:1976js} where, instead of considering complex deformation $\textrm{P}\T$ of the projective twistor space $\mathbb{PT}$ preserving the Poisson structure along the holomorphic fibration $\mathbb{PT} \to \mathbb{C}\textrm{P}^1$, we considered a complex deformation $\T$ of the twistor space $\mathbb{T}$: The crucial difference lies in the relaxed requirement \ref{item euler not holom} on the Euler vector field, which generates the quotient $\mathbb{T} \to \mathbb{PT}$. As a result, the complex deformations that we considered are not deformations of $\PT$ but they are rigid enough to be interpreted as being generated by higher-spin fields, see Proposition \ref{Prop decomposition of deformations}. In analogy with Penrose's classical result, the higher-spin space $\HST$ is then defined as the space of holomorphic embeddings $\mathbb{C}^2_0 \hookrightarrow \T$.

Under these general assumptions, negative helicity higher-spin fields will source the deformation and, thus, one is forced to allow for singularities in the embeddings $\mathbb{C}^2_0 \hookrightarrow \T$. This is consistent with the fact that, as discussed in section \ref{section: Plebanski potentials}, such singularities are necessary to obtain the Plebanski formulation of (the Poisson contraction of) chiral higher-spin gravity. Therefore, singularities in the embedding of the spacetime point seem to be a key feature of the presence of negative helicity fields, a point which was not apparent from \cite{Mason:2025pbz}.

On the other hand, the appearance of negative helicity fields, and related singularities of the planes, can be prevented by the extra assumption that the complex deformations are bounded at the origin \ref{item bounded deformation}, which is the main setup of this work. Under this extra assumption, we proved (see Theorem \ref{theorem existence of solutions}) that the moduli space $\HST$ of holomorphic embeddings $\mathbb{C}^2_0 \hookrightarrow \T$ is a infinite dimensional complex manifold and canonically projects $\HST \to \M$ on a usual self-dual four-dimensional  spacetime. This leads to a generalization of the nonlinear graviton theorem, summarized in Theorems $\ref{theorem NLG}$  and $\ref{theorem NLG with gauge}$, identifying complex deformations of twistor space with solutions of higher-spin self-dual gravity. An novel feature of the construction is that higher-spin symmetries are here geometrically realized as choices of embeddings $\M\hookrightarrow \HST$ of spacetime $\M$ into the higher-spin space.\\

There are two main outcomes of the present article, whose investigation is left for future work. On the one hand, our twistor construction naturally produces a notion of higher-spin geometry, in the form of an infinite dimensional complex manifold equipped with an infinite number of higher-spin frames, the higher-spin symmetry appearing as a result of the restriction of this infinite dimensional geometry to a four-dimensional slice. On the other hand, it naturally suggests an extension of our results to include negative helicity interactions. This is because the boundness assumption \ref{item bounded deformation} here seems to merely play the role of excluding such interactions. As we already emphasized, however, such negative helicity interactions have to come, in the present setting, hand-in-hand with singularities in the holomorphic embeddings $\mathbb{C}^2_0 \hookrightarrow \T$. The geometrical understanding, as well as the physical interpretation, of such singularities is likely to come with interesting subtleties that we leave for future work. Finally, Penrose's original result \cite{Penrose:1972ia} was intrinsically an equivalence between conformal self-dual metric and deformation of (projective) twistor space, with Einstein metrics coming as a further restriction. In contrast, for conformal higher-spin theory 
\cite{Segal:2002gd,Tseytlin:2002gz}, 
while related twistor actions have been written some time ago already \cite{Hahnel:2016ihf,Adamo:2016ple}, 
the question as to whether or not there exists a related nonlinear graviton theorem, and how it would relate to the present results, is still open.

\section*{Acknowledgements}

N.B. is grateful to discussions with Thomas Basile, Carlo Iazeolla and 
Evgeny Skvortsov. His work was partially supported by the F.R.S.-FNRS PDR 
grant number T.0047.24. 
The authors would like to thank the Institut Denis Poisson (U. Tours) 
and UMONS for hospitality. N.P. is a FNRS grantee
of the Fund for Scientific Research — FNRS, Belgium.
N.P. wants to thank the SophinaBo\"el foundation and the Mathematical Institute
of Oxford for enabling an internship where part of this work was accomplished. 

The authors are also happy to acknowledge that part of the project was 
carried out during the meetings \href{https://web.umons.ac.be/pucg/en/event/workshop-twistors-and-higher-spins/}{\tt Twistors and Higher Spins} and 
\href{https://web.umons.ac.be/pucg/en/event/workshop-conformal-higher-spins-twistors-and-boundary-calculus/}{\tt Conformal higher spins, twistors and 
boundary calculus} that took place at UMONS on July 2024 and 2025 respectively, 
and would like to thanks the organisers for the invitation.

\appendix 
\section{higher-spin self-dual gravity fields in terms of their Plebanski potentials}

In this appendix we establish the explicit relation 
between the Plebanski potentials from Section \ref{section:  Plebanski potentials} 
and the spacetime fields deduced from twistor space. 
We also compute the Lax pairs and the dual frames from these Plebanski 
potentials. 
We also check that we consistently recover the correct fields 
equations from this relation. 
The relationship is given by putting the HS fields in the gauge :

 \begin{equation}
{\rm d}x^{\alpha\dot\alpha}A_{\alpha\dot\alpha}^{\beta(k)}
= -2\,o_\alpha {\rm d}x^{\alpha\dot\alpha}
o^\gamma\partial_{\gamma\dot\alpha}\theta_{k+2} o^{\beta(k)}\, , 
\qquad o_\alpha x^{\alpha\dot\alpha}= z^{\dot \alpha}  \, , \quad 
o^{\alpha}\partial_{\alpha\dot\alpha}=-\partial/\partial w^{\dot\alpha}\; .
 \end{equation}
Where $o_\alpha $ is a choice of reference spinor and is part of the gauge choice. 
In linear theory the existence of such a gauge potential 
was proved in \cite{Penrose:1965am} and follows from the Penrose 
transform.  
With this
\begin{equation}
    A(x,\lambda)=A_0 + 2\,{\rm d}z^{\dot\alpha} \sum_{k} \lambda_2^k 
    \theta_{k \dot\alpha} \, ,\qquad \partial_{w^{\dot\alpha}}  
    \theta_k =\theta_{k \dot\alpha}\, , 
    \qquad \lambda_2=\lambda_\alpha o^\alpha \;.
\end{equation}
Here $A_0 := x^{\alpha}{}_{\dot \alpha }{\rm d}x^{\dot\alpha \alpha}\lambda_{\alpha(2)}$
is the potential of a flat background: one can check that 
\begin{align}
	A_0 = -\lambda_1^2 z_{\dot\alpha}{\rm d}z^{\dot\alpha} 
    - \lambda_1\lambda_2(z_{\dot\alpha}{\rm d}w^{\dot\alpha} 
    - w_{\dot\alpha}{\rm d}z^{\dot\alpha}) 
    + \lambda_2^2 w_{\dot\alpha} {\rm d}w^{\dot\alpha}\;,
\end{align} 
and thus, 
\begin{equation}\label{eq dA0}
    {\rm d} A_0= \lambda_1^2{\rm d}z^{\dot\alpha}\wedge {\rm d}z_{\dot \alpha} 
    + 2\lambda_2\lambda_1 {\rm d} z^{\dot\alpha} \wedge {\rm d}w_{\dot\alpha}
    + \lambda_2^2 {\rm d} w^{\dot \alpha}\wedge {\rm d} w_{\dot\alpha}\;,
\end{equation}
which automatically satisfies ${\rm d}A_0\wedge {\rm d}A_0 = 0\,$. 

It is also possible to identify the Lax pair (i.e. the vector fields generating the distribution \eqref{Correspondence space: lax pair}) as
\begin{equation}
    L_{\dot\alpha}= \lambda_1 \partial_{w^{\dot\alpha}} -\lambda_2 \left(\partial_{z^{\dot\alpha}} -2\,\sum_k \lambda_2^{k-2} \theta_{k\dot \alpha}{}^{\dot\beta}\partial_{w^{\dot\beta}}\right)\;,
\end{equation}
and the frame, that annihilates the Lax pair, 
$E^{\dot\alpha}(L_{\dot\beta}) = 0\,$, is 
\begin{equation}
	E^{\dot\alpha} = \lambda_1 {\rm d}z^{\dot\alpha}
	+ \lambda_2 \left({\rm d}w^{\dot\alpha} 
	+ 2\,\sum_k \lambda_2^{k-2}
    \theta_k^{\dot\alpha}{}_{\dot\beta}\,dz^{\dot\beta}\right)\;.
\end{equation}
With these results, one can check that imposing 
${\rm d} A \wedge {\rm d}A = 0$ gives the general spin $s$ equations 
(here for even spins)
\begin{align}
	{\rm d}A \wedge {\rm d}A &= 
	-d^2 z\wedge d^2w\; \sum_{n \in \mathbb{N}}\lambda_2^{2n+2}\left(\square \theta_{2n}
	+  \sum_{\overset{p,q st.}{p+q = n+1}}
	\theta_{2p}\ \overset{\leftrightarrow}{P}^2 \ \theta_{2q}  \right) 
\end{align}
with
\begin{align}
	\square & := \epsilon^{\alpha\beta}\epsilon^{\dot\alpha\dot\beta}\frac{\partial}{\partial x^{\dot\alpha\alpha}}\frac{\partial}{\partial x^{\dot\beta\beta}} = 2 \epsilon^{\dot\alpha\dot\beta}
	\frac{\partial}{\partial z^{\dot\alpha}}\frac{\partial}{\partial w^{\dot\beta}} \\
	\text{ and }\quad f\ \overset{\leftrightarrow}{P}^2 g
	& := \epsilon^{\dot\alpha\dot\gamma}\epsilon^{\dot\beta\dot\delta}\partial_{w^{\dot\alpha}}\partial_{w^{\dot\beta}}f
	\partial_{w^{\dot\gamma}}\partial_{w^{\dot\delta}}g.
\end{align}

\section{Fall-off conditions at the origin}\label{Apdx: Euclidean fields}

In this section we expand on the fall-off conditions at the origin of twistor space for the Penrose representatives $\h_{(k)} \in H^{(0,1)}(\mathbb{PT},\mathcal{O}(k))$. 

The discussion is mostly explicitly done in Euclidean signature for which
\begin{equation}
    x^{\alpha\dot\alpha} = \frac{\mu^{\dot\alpha}\hat\lambda^\alpha- \hat\mu^{\dot\alpha}\lambda ^\alpha}{\langle \lambda\hat\lambda\rangle}\;,
\end{equation}
where $\langle \lambda\hat\lambda\rangle = \lambda_{\alpha}\hat\lambda^{\alpha}$ is the $SU(2)$ inner-product corresponding to the choice of Euclidean reality condition. 

Making use of Woodhouse gauge (i.e. harmonic gauge on the $\lambda$-sphere) \cite{Woodhouse:1985id}, one has, for $ n\geq 0$:
\begin{align}
    \h_{(n-1)}&= \frac{{\rm d}x^{\dot\alpha \beta}\hat \lambda_\beta}{\langle \lambda\hat\lambda\rangle}\phi_{\dot\alpha \alpha(n)}\lambda^{\alpha(n)}\; ,  \\
    \h_{(-n-2)}&=\frac{1}{\langle \lambda\hat\lambda\rangle^{n+2}}\left( \hat{\lambda}^{\gamma}{\rm d}\hat\lambda_{\gamma} \psi_{\alpha(n)}\hat\lambda^{\alpha(n)} + {\rm d}x^{\dot\alpha \gamma}\hat\lambda_{\gamma} \hat\lambda^\beta \nabla_{\dot\alpha\beta}\psi_{\alpha(n)} \hat\lambda^{\alpha(n)}\right)\;,
\end{align}
where $\phi_{\dot\alpha \alpha(n)}(x)$, and $\psi_{\alpha(n)}(x)$ are the corresponding fields on spacetime. 

In order for $\phi$ and $\psi$ to be non-trivial solutions of the zero-rest mass equation on $\mathbb{R}^4$ they will have to diverge as $x\rightarrow\infty$: as a result,  when $\mu$ is kept fixed, one expects singularity for $\h_{(k)}$ along $\lambda=0$.  

Because of this, the Penrose representatives will in general be singular at the origin. However, if one restricts to $|\mu|< R|\lambda|$, i.e. to points $x$ in some open ball around $x=0$, then $\h_{(k)}$ will remain bounded.

\section{Infinite-dimensional higher-spin spacetime in Vasiliev's theory}
\label{Appendix:Vasiliev}

In this appendix we review the proposal of \cite{Sezgin:2011hq}, 
in the context of the bosonic models describing an infinite tower of 
integer-spin gauge fields propagating on AdS$_4$. Along those lines, an 
action principle for Vasiliev's equations and various extensions thereof 
were given in \cite{Boulanger:2011dd,Boulanger:2015kfa}. We first 
very briefly summarise some key aspects of Vasiliev's equations.
Detailed, technical expositions can be found e.g. in 
\cite{Vasiliev:1999ba,Didenko:2014dwa}.
Here we mostly follow \cite{Sezgin:2011hq,Boulanger:2011dd}.

\paragraph{Vasiliev's bosonic models.} 
The fields, a zero-form and a one-form $(\hat{\Phi},\hat{A})$,
entering Vasiliev's description \cite{Vasiliev:1990en}, 
are elements of a noncommutative, 
associative algebra 
\begin{equation}
 \hat{\cal A}=\Omega({\cal B})\otimes {\cal Y}   
\end{equation}
where $\Omega({\cal B})$ is the space of differential forms on a 
noncommutative base manifold ${\cal B}$,
and ${\cal Y}$ is an associative, noncommutative higher-spin algebra.
The base manifold ${\cal B}$ receives a trivial bundle structure 
\begin{align}
{\cal B} = {\cal Z}_4 \times {\cal M}_{HS}    
\end{align}
where ${\cal Z}_4$ is a 4-dimensional noncommutative 
space with topology\footnote{Other topologies can be taken, 
like for example $S^2\times S^2$. 
See \cite{Iazeolla:2007wt,Boulanger:2015kfa} 
for detailed discussions of reality conditions and topologies 
to be imposed on ${\cal Z}_4$.} ${\mathbb C}^4$, 
whose coordinates are usually presented as a pair 
of $SL(2,\mathbb{C})$ doublets 
$Z^A = ( z^{\alpha} , \bar{z}^{\dot{\alpha}} )\,$,  
and ${\cal M}_{HS}$ is a commuting manifold with local coordinates 
$X^{M}$ that contains the usual four-dimensional spacetime with 
local coordinates $x^\mu\,$, $\mu = 0, \ldots, 3$.
The exterior derivative on $\cal B$ is locally given by 
${\rm d} = {\rm d}X^{\underline{M}}\, {\partial}_{\underline{M}} 
+ {\rm d}z^\alpha \frac{\partial}{\partial z^\alpha}
+ {\rm d}\bar z^{\dot\alpha} \frac{\partial}{ \partial \bar z^{\dot\alpha}}\,$.
The associative higher-spin algebra ${\cal Y}$ that plays the role 
of the fiber in the trivial bundle ${\cal B}\times {\cal Y}$ 
is the algebra of polynomials\footnote{That is often extended 
to functions or even distributions in the variables $Y^A$.} 
in the four variables 
$Y^A = ( y^{\alpha} , \bar{y}^{\dot{\alpha}})$,  
equipped with the associative Moyal star product on 
each doublet of variables. We denote by the symbol $\star$ the associative, 
noncommutative product on $\hat {\cal A}\,$.
In the Weyl ordering scheme, the variables $Y^A$ and $Z^A$ 
are decoupled, $[Y^A,Z^B]_\star = 0\,$, as it should for a clear 
separation between the base manifold $\cal B$ and the fiber $\cal Y$. 
The Vasiliev equations really are operator-valued equations, and various 
ordering schemes can be adopted for the symbols of the operators.

The associative algebra $\hat{\cal A}= \Omega({\cal B})\otimes {\cal Y}$
admits two automorphisms $(\pi,\bar{\pi})$ that commute with 
de Rham's differential on $\cal B$ and whose actions are defined by
\begin{align}
\pi(X^{M} , y^\alpha, \bar{y}^{\dot\alpha}; z^\alpha, \bar{z}^{\dot\alpha})
& = (X^{M} , -y^\alpha, \bar{y}^{\dot\alpha}; -z^\alpha, \bar{z}^{\dot\alpha})\;,
\\
\bar\pi(X^{M} , y^\alpha, \bar{y}^{\dot\alpha}; 
z^\alpha, \bar{z}^{\dot\alpha})
& = (X^{M} , y^\alpha, -\bar{y}^{\dot\alpha}; z^\alpha, -\bar{z}^{\dot\alpha})\;.
\end{align}
The bosonic, higher-spin Vasiliev models \cite{Vasiliev:1990en} 
are selected by imposing the two conditions 
\begin{align}\label{bosonic_model}
\pi\bar\pi(\widehat{\Phi},\widehat{A} ) =  (\widehat{\Phi},\widehat{A} )\;.
\end{align}
Vasiliev's models contain a scalar field 
$\phi(x^\mu)$ on ${\cal M}_{HS}$,
as a distinguished component (and projection) of $\hat{\Phi}$. 
Depending on whether one
assigns to $\phi$ an intrinsic spacetime parity $\pm 1$ and demanding 
that the nonlinear field equations are parity invariant, one obtains 
two bosonic models for which the field equations read
\begin{align}\label{VE}
\quad \hat{F} := {\rm d} \hat{A} + \hat{A}\star \hat{A} 
= \frac{i}{4}\,(b\,d^2z \;\hat{\Phi}\,\star \kappa 
+ \bar b\,d^2\bar z \;\hat{\Phi}\,\star \bar\kappa)  \;,\quad 
\hat{D}\hat{\Phi} := {\rm d}\hat{\Phi} + [\hat{A},\hat{\Phi}]^\star_\pi = 0\;, 
\end{align}
where in Weyl order\footnote{In the original work \cite{Vasiliev:1990en}, 
a different ordering scheme was adopted for the symbols of the operators. 
In the normal ordering used in \cite{Vasiliev:1990en}, one has the 
Gaussian-like expressions $\kappa = \exp{(i y^\alpha z_\alpha)}$ and 
$\bar\kappa = \exp{(i \bar y^{\dot\alpha} \bar z_{\dot\alpha})} \,$.}, 
one has
$\kappa = 2\pi \,\delta^2(y)\star 2\pi\, \delta^2(z)\,$,
$\bar\kappa = 2\pi\, \delta^2(\bar y)\star 2\pi \,\delta^2(\bar z)$,
and 
\begin{align}
&{\rm {Type~A}}\;({\rm parity~even~scalar}):\qquad b=1\;,\\
&{\rm {Type~B}}\;({\rm parity~odd~scalar}):\qquad b=i\;.
\end{align}

\paragraph{Gauge algebra.} 

Upon introducing the anti-automorphism $\tau$ that acts as 
\begin{align}
\tau(X^{M} , y^\alpha, \bar{y}^{\dot\alpha}; 
z^\alpha, \bar{z}^{\dot\alpha})
= (X^{M} , i y^\alpha, i \bar{y}^{\dot\alpha}; -i z^\alpha, 
-i \bar{z}^{\dot\alpha})\;,
\end{align}
one may further restrict the type A and B models to 
what are called the two \emph{minimal} models that propagate 
only the even-spin fields around (A)dS$_4$ 
background\footnote{The reality conditions that are necessary to 
select the de Sitter (dS$_4$) or the anti-de Sitter (AdS$_4$) backgrounds
can be found e.g. in \cite{Iazeolla:2007wt}, as well as the 
reality conditions for other -- but non-flat -- backgrounds.}:
\begin{align}\label{minimal_bosonic_constraint}
    \tau(\hat{A},\hat{\Phi}) = (-\hat{A},\pi(\hat\Phi))\;.
\end{align}
It can be seen that the above constraints imply \eqref{bosonic_model}.

Vasiliev's nonlinear equations are invariant under the following gauge 
transformations:
\begin{align}
    \delta_{\hat{\epsilon}}\hat{A} = \hat D \hat{\epsilon}
    ={\rm d}\hat\epsilon+[\hat{A},\hat{\epsilon}]_\star\;,
    \quad \delta_{\hat{\epsilon}}\hat{\Phi} 
    = -[\hat{\epsilon} , \hat{\Phi}]^\star_\pi\;, 
\end{align}
for the locally-defined zero-form infinitesimal gauge parameters 
$\hat{\epsilon}(x,Y,Z)$
that obey the same kinematic constraints as the one-form field 
$\hat{A}(x,Y,Z)$, i.e., $\tau(\hat{\epsilon}) = - \hat{\epsilon}\,$. 
The closure of the gauge transformations reads
\begin{align}
    [\delta_{{\hat{\epsilon}}_1}, \delta_{{\hat{\epsilon}}_2}]
    = \delta_{{\hat{\epsilon}}_{_{12}}}\;,
    \quad {\hat{\epsilon}}_{_{12}} := 
    [{{\hat{\epsilon}}_1, {\hat{\epsilon}}_2}]_\star\;,
\end{align}
thereby def ining the complex Lie algebra 
\begin{align}
\hat{\mathfrak{g}}=\mathbb{C}\otimes_\mathbb{R}\hat{\mathfrak{hs}}(4) 
= \left\{ \hat{P}(Y,Z)~:~\tau(\hat{P}) = - \hat{P}
\right\}\;,\qquad {\rm ad}_{\hat{P}_1}(\hat{P}_2) = [\hat{P}_1,\hat{P}_2]_\star\;,
\end{align}
given by the complexification of the real higher-spin Lie algebra 
$\hat{\mathfrak{hs}}(4)$ defined in \cite{Sezgin:2011hq}.
For an algebra that would contain all the generators of integer 
spin fields, one needs to replace the condition  
$\tau(\hat{P}) = - \hat{P}$ by the weaker 
$\pi\bar{\pi}(\hat{P}) = \hat{P}\,$.
We do not intend to review the various possible reality conditions 
compatible with Vasiliev's equations, and only mention that the latter admit 
\cite{Iazeolla:2007wt} the symmetric spaces 
\[ \{S^4 , H_4=SO(4,1)/SO(4), dS_4 , AdS_4 , H_{3,2}=SO(3,2)/SO(2,2)\}\]
as exact, maximally-symmetric solutions, that exclude Minkowski's spacetime.
In the following, we shall consider the complex algebra 
$\hat{\mathfrak{g}}$ when the authors of \cite{Sezgin:2011hq} considered 
the real algebra $\hat{\mathfrak{hs}}(4)$.

\paragraph{Cartan's formulation.}

In the paper \cite{Sezgin:2011hq}, the authors adopt a description of 
Vasiliev's theory in terms of a Cartan geometry (although no definition
of higher-spin groups are attempted, the higher-spin algebras being 
infinite-dimensional, and not of the Kac-Moody type), 
with a connection one-form taking its values in the 
Lie algebra $\hat{\mathfrak{g}}$, and where the isotropy Lie algebra  
$\hat{\mathfrak{h}}\subseteq\hat{\mathfrak{g}}$ is essentially\footnote{The 
Lorentz algebra that naturally acts on the indices of the physical 
spacetime fields is not canonically embedded into $\hat{\mathfrak{g}}$, 
therefore some care \cite{Vasiliev:1999ba,Sezgin:2002ru} must 
be taken in defining it, as we shall briefly review.} 
given by the $\pi$-even subalgebra 
$\hat{\mathfrak{g}}_+ \subset \hat{\mathfrak{g}}\,$.

In order to give a manifestly Lorentz-covariant description
of Vasiliev's equations, Vasiliev \cite{Vasiliev:1999ba} 
(see \cite{Sezgin:2002ru} for further discussions) introduced 
$\hat{S}_\alpha = z_\alpha - 2i \hat{A}_\alpha\,$, 
$\hat{S}_{\dot\alpha} = \bar z_{\dot\alpha} - 2i \hat{A}_{\dot\alpha}\,$,
and the field-dependent Lorentz generators 
$\hat{M}_{\alpha\beta}:=\hat{M}_{\alpha\beta}^{(0)}+
\hat{M}_{\alpha\beta}^{(S)}$ (similarly in the dotted indices)
where 
\begin{align}
\hat{M}_{\alpha\beta}^{(0)}=y_{(\alpha}\star y_{\beta)}
- z_{(\alpha}\star z_{\beta)}\;,\quad
\hat{M}_{\alpha\beta}^{(S)}=\hat{S}_{(\alpha}\star \hat{S}_{\beta)}\;.
\end{align}
One then introduces the following one-form on 
${\cal M}_{HS}$: 
\begin{align}
\hat{W} = dX^{{M}}\,\hat{W}_{{M}} 
:= dX^{{M}}\,\hat{A}_{{M}}
-\frac{1}{4i}dX^{{M}}\,(\omega_{{M}}^{\alpha\beta}\hat{M}_{\alpha\beta}
+\bar\omega_{{M}}^{\dot\alpha\dot\beta}
\hat{M}_{\dot\alpha\dot\beta})\;.
\end{align}
In terms of this redefined one-form on ${\cal M}$, the horizontal part 
of the Vasiliev field equations now take a manifestly Lorentz-covariant 
form, and one can \cite{Sezgin:2011hq} always  
impose the condition that $\hat{W}$ 
has a vanishing components along the $Z$-independent 
Lorentz generators $y_{(\alpha}\star\, y_{\beta)}\,$
and $\bar{y}_{(\alpha}\star\, \bar{y}_{\beta)}\,$, 
that instead appear in the Lorentz connection
\begin{align}
\omega=\frac{1}{4i}dx^{{M}}
(\omega_{{M}}^{\alpha\beta}\hat{M}_{\alpha\beta}
+\bar\omega_{{M}}^{\dot\alpha\dot\beta}\hat{M}_{\dot\alpha\dot\beta})\;.
\end{align}
The one-form gauge field $\hat{W}+\omega$ takes its values in 
$\hat{\mathfrak{g}}$, and a Cartan-like formulation of Vasiliev's 
minimal models consists of choosing a structure, or isotropy algebra
\begin{align}
\hat{\mathfrak{h}}\subseteq \hat{\mathfrak{g}}\;.
\end{align}
Then, in terms of such a structure algebra, the two projections 
$\hat{\Gamma}:=\Pi_{\hat{\mathfrak{h}}}(\hat{W}+\omega)$ and 
$\hat{E}:=(\mathtt{id}-\Pi_{\hat{\mathfrak{h}}})(\hat{W}+\omega)$ decompose the 
Cartan-like, $\hat{\mathfrak{g}}$-valued connection one-form into
\begin{align}
\hat{W}+\omega=:\hat\varpi = \hat{E} + \hat{\Gamma}
\end{align}
where the component $\hat{E}$ formally plays the role 
of generalised soldering one-form, while $\hat{\Gamma}$ 
takes its values in the structure algebra, playing the role 
of generalised Lorentz connection.

As explained in \cite{Sezgin:2011hq}, each choice of structure subalgebra 
$\hat{\mathfrak{h}}\subset \hat{\mathfrak{g}}$ 
corresponds to a specific phase of the theory, and 
the authors of \cite{Sezgin:2011hq} give several examples, among which 
a broken phase with what they call a chiral higher-spin structure 
algebra 
\begin{align}
\hat{\mathfrak{hl}}(2,\mathbb{C}) = \left\{
\widehat{P}(y,z)+\widehat{P}(\bar{y},\bar{z})\right\}\cap \hat{\mathfrak{g}}_+\;,
\end{align}
that is, all polynomials in the $\pi$-even algebra $\hat{\mathfrak{g}}_+$
that are purely holomorphic or anti-holomorphic. The author explain that 
(we quote), ``the geometrical significance of this choice of structure 
algebra [...] remains to be investigated.''  
In fact, one could further reduce the above structure algebra 
and consider what could be called a self-dual algebra  
\begin{align}
\hat{\mathfrak{l}}_+ := 
\left\{ \hat{P}(y,z)~:~\tau(\hat{P}) = - \hat{P}\right\}\cap \hat{\mathfrak{g}}_+\;.
\end{align}
The structure subalgebra that was mainly considered in \cite{Sezgin:2011hq}
is $\hat{\mathfrak{h}}=\hat{\mathfrak{g}}_+ = \frac{1}{2}\,(\mathtt{id} + \pi) \hat{\mathfrak{g}}\,$,
so that the complement 
$\hat{\mathfrak{p}} = \frac{1}{2}\,(\mathtt{id} - \pi) \hat{\mathfrak{g}}$ 
in the vector space decomposition 
$\hat{\mathfrak{g}}=\hat{\mathfrak{p}} + \hat{\mathfrak{h}}$
plays a role analogous to the tangent space at any point of the symmetric 
space $G/H$, in case one could make sense of the exponentiations of 
$\hat{\mathfrak{g}}$ and $\hat{\mathfrak{h}}$ into groups $G$ and $H$, 
respectively.

\paragraph{Symmetric higher-spin spacetime.}

We have the symmetric-space decomposition 
$\hat{\mathfrak{g}}=\hat{\mathfrak{h}}+\hat{\mathfrak{p}}\,$ 
with the involutive automorphism $\pi$ of $\hat{\mathfrak{g}}$ 
such that $\pi = \mathtt{id}_{\hat{\mathfrak{h}}}\oplus 
(-\mathtt{id}_{\hat{\mathfrak{p}}})$, 
the soldering one-form $\hat{E} = \hat{W}^-=\frac{1}{2}(\mathtt{id} - \pi)\hat{W}$ 
and $\hat{\Gamma}= \hat{W}^+ + \omega\,$. Similarly, the gauge parameters  
$\hat{\epsilon}\in \hat{\mathfrak{g}}$
are decomposed into 
$\hat{\xi} := \frac{1}{2}(\mathtt{id} - \pi)\hat{\epsilon}$ and 
$\hat{\Lambda} = \frac{1}{2}(\mathtt{id} + \pi)\hat{\epsilon}\oplus \Lambda$, 
where the parameter $\Lambda$ is responsible for the local Lorentz 
transformations.  
The broken gauge symmetries that in \cite{Sezgin:2011hq} are referred 
to as \emph{generalised local translations}, act as 
\begin{align}
    \delta_{\hat{\xi}}\hat{\Gamma} = [\hat{E},\hat{\xi}]_\star\;,
    \quad \delta_{\hat{\xi}}\hat{E} = \hat{\nabla}\hat{\xi}\;,
    \quad \delta_{\hat{\xi}}\omega = 0\;,
\end{align}
where $\hat{\nabla} = \nabla + {\rm ad}_{\hat{W}^+}$ with 
$\nabla = {\rm d} + \rho(\omega)$.
The fields $(\hat{E},\hat{\xi})$ are sections of the principal 
$\hat{\mathfrak{g}}_+ \oplus \mathfrak{sl}(2,\mathbb{C})$ bundle,
and the symmetric space picture suggests to identify the tangent 
space of ${\cal M}_{HS}$ at any point with 
$\hat{\mathfrak{g}}/\hat{\mathfrak{g}}_+\,$.
In other words, in this approach the commutative spacetime manifold 
${\cal M}_{HS}$ is infinite-dimensional\footnote{Attempts to define 
higher-spin diffeomorphisms on a finite-dimensional manifold
can be found in \cite{Bekaert:2021sfc}.}, 
and the soldering one-form 
$\hat{E}$ on ${\cal M}_{HS}$ gives rise to the frame fields 
$\hat{E}_{{M}}=\hat{E}_{{M}}{}^{{A}}\,\hat{P}_{{A}}\,$, 
$\pi(\hat{P}_{{A}})=-\hat{P}_{{A}}\,$, where 
$\hat{P}_{{A}}$ are the coset generators. 
The generalised frame fields are supposed to be invertible, with 
$\hat{E}^{{M}}{}_{{A}}$ the (formal) inverse of 
$\hat{E}_{{M}}{}^{{A}}$.
The local translations with gauge parameters 
$\hat{\xi}=\hat{\xi}^{{A}}$
are identified with infinitesimal diffeomorphisms with globally 
defined vector fields 
$\xi^{{M}} = \hat{E}^{{M}}{}_{{A}}\,
\hat{\xi}^{{A}}\,$, combined with local generalised 
Lorentz transformations \cite{Sezgin:2011hq}.

Like in the superspace formulation of supersymmetric field theories, 
the Vasiliev equations can be studied perturbatively around the 
``body'' where the (here Grassmann-even and noncommutative) 
variables $Z^A$ are set to zero, further\footnote{The analysis of 
\cite{Boulanger:2015ova} shows that the reduction of Vasiliev's equations 
to $\cal M=$AdS$_4$ is only formal, in the sense that it produces 
highly non-local equations on AdS$_4$.} 
reducing the quasi-free differential algebra on $\cal B$ to a  
free differential algebra on a ${\cal M}_{HS}$.

In the reduction of the equations from $\cal B$ to 
${\cal M}$, one therefore encounters 
infinitely many commuting coordinates $X^M$ in bijection with 
the $\pi$-odd generators 
$P_A = \hat{P}_{{A}}\vert_{Z=0}$ of the 
algebra $\mathfrak{g}=\hat{\mathfrak{g}}\vert_{Z=0}$.
The analysis \cite{Sezgin:2011hq} of the resulting field 
equations for the master zero-form $\Phi(X|Y)=\hat\Phi(X,Z=0|Y)$ 
on ${\cal M}_{HS}$ involves introducing a decomposition 
$P_A=\{P_{A_\ell}\}_{\ell = 0}^\infty$ where 
\begin{equation}
P_{A_\ell} = \{ P_{a(2\ell+1),b_{2k}}\}_{k=0}^\ell = 
\{ M_{a_1b_1}\star\ldots \star M_{a_{2k}b_{2k}}\star
P_{a_{2k+1}}\star\ldots\star P_{a_{2\ell+1}}\}_{k=0}^\ell\;, 
\end{equation}
and where $a,b, \ldots$ take the values $0,1,2,3$, 
$\{M_{ab},P_a\}$ representing in $\cal Y$ the generators of $\mathfrak{so}(2,3)$. 
The tensor $P_{a(2\ell+1),b_{2k}}$ is Lorentz-irreducible of 
type $(2\ell+1,2k)\,$. For more details about the star-product 
realisation of the higher-spin algebra, see the very 
complete reference \cite{Iazeolla:2008ix}.
Accordingly, there is a notion of level decomposition of the 
commuting coordinates 
\begin{equation}
X^{M_\ell} = \{ X^{\mu(2\ell+1),\nu_{2k}}\}_{k=0}^\ell\;, 
\quad \ell\in \mathbb{N}\;,
\end{equation}
where the usual, local coordinates $x^\mu$ are identified with 
$X^{M_\ell}\vert_{\ell=0}\,$:
\begin{equation}
    x^\mu = X^{M_0}\;, 
\end{equation}
and at each level $\ell\in\mathbb{N}$ the $k=0$ coordinate is a 
completely symmetric and traceless tensor $X^{\mu_1\ldots\mu_{2\ell+1}}$.
These coordinates $X^{M_\ell}$ are the counterpart, for Vasiliev's equations, 
to the coordinates $x^I$ in our Theorem \ref{theorem existence of solutions} 
above in the context of SD-HSGR. 
Note also the existence of the generalised frame field on ${\cal M}_{HS}$ 
\begin{equation}
 E_M{}^A(X) := \hat{E}_M{}^A(X,Z=0)   
\end{equation}
that plays a role analogous to the frame 
\eqref{HS frames} in our twistor 
reformulation of SD-HSGR. 

The authors of \cite{Sezgin:2011hq} further assume the existence 
of a triangular gauge for the horizontal one-form ${\rm d}X^M\,E_M{}^A$ 
such that $E_{M_\ell}{}^{A_{\ell'}}\vert_{X^{M_{\ell''}}=0} = 0$ 
for $\ell>\ell'$ and for all $\ell''>0\,$, exactly in analogy to 
what happens for the superframe in the superspace formulation of 
supergravity, see in particular the equation (17.2) of \cite{Wess:1992cp}.

At level $\ell$ in the decomposition of the equation for the master 
zero-form $\Phi=\hat{\Phi}\vert_{Z=0}$, the equation reads 
$\nabla_{a(2\ell+1)}(\Gamma, \omega)\Phi + \{P_{a(2\ell+1)},\Phi\}_\star =0\,$, 
where 
\begin{equation}
P_{a(2\ell+1)} = P_{\{a_1}\star\ldots\star P_{a_{2\ell+1}\}}    
\end{equation}
is the translation-like generator associated with the coordinate 
$X^{\mu_1\ldots \mu_{2\ell+1}}\,$, and the curly brackets around the 
Lorentz indices $a_1, \ldots, a_{2\ell+1}$ indicate symmetrisation and 
traceless projection. The authors of \cite{Sezgin:2011hq} then argue that 
the generalized Lorentz-covariant derivative 
$\nabla_{a_{2\ell+1}}(\Gamma, \omega)$ factorises, for example on the 
scalar field $\phi=\Phi\vert_{Y=0}$, into multiple vectorial derivatives: 
$\nabla_{a_{2\ell+1}}(\Gamma, \omega)\phi \propto$ 
$\nabla_{\{a_{1}}(\Gamma, \omega)\ldots \nabla_{a_{2\ell+1}\}}(\Gamma, \omega)\phi\,$,
so that the dependence of the physical fields (here, the scalar) on the 
higher coordinates is, on-shell, expressed in terms of their dependence 
on the usual coordinate $x^\mu$, similarly to what happens for the dependence 
of superfields on the fermionic coordinates, as reviewed e.g. in the 
textbook \cite{Wess:1992cp}.

\providecommand{\href}[2]{#2}\begingroup\raggedright\endgroup

\end{document}